\documentclass[twoside]{aiml22}

\usepackage{aiml22macro}

\usepackage{graphicx}
\usepackage{amsmath}
\usepackage{amssymb}

\usepackage{tikz}
\usepackage{bussproofs}
\usepackage{amssymb}
\usepackage{multirow}
\usepackage{stmaryrd}
\usepackage{mathrsfs}
\usepackage{lipsum} 

\usepackage{virginialake}
\usepackage{adjustbox}

\newcommand{\neigh}{neighbourhood}

\newcommand{\CPN}{CPN}
\newcommand{\cond}{>}
\newcommand{\less}{\preccurlyeq}

\newcommand{\ax}{\AxiomC}
\newcommand{\uinf}{\UnaryInfC}

\newcommand{\disp}{\DisplayProof}
\newcommand{\G}{\Gamma}
\newcommand{\D}{\Delta}
\newcommand{\seq}{\Rightarrow}
\newcommand{\tr}{\triangleleft}

\newcommand{\hyp}{\mid}
\newcommand{\hG}{\mathcal{G}}
\newcommand{\hH}{\mathcal{H}}

\newcommand{\fint}[1]{i(#1)}
\newcommand{\Deltaset}[1]{{#1}^\D}
\newcommand{\SDelta}{\Deltaset{\Sigma}}

\newcommand{\bl}[2]{[#1 \tr #2]}

\newcommand{\Vd}{\Vdash}
\newcommand{\vd}{\vdash}

\newcommand{\AND}{\bigwedge}
\newcommand{\OR}{\bigvee}

\newcommand{\lan}{\mathscr L}
\newcommand{\atm}{Atm}
\newcommand{\pow}{\mathcal{P}}
\newcommand{\set}[1]{set(#1)}

\newcommand{\W}{ W}
\newcommand{\N}{ N}
\newcommand{\V}{ V}
\newcommand{\M}{\mathcal M}

\newcommand{\R}{\mathcal R}
\newcommand{\C}{\mathcal C} 

\newcommand{\Wl}{W_{\Logic}}
\newcommand{\Nl}{N_{\Logic}}
\newcommand{\Vl}{V_{\Logic}}
\newcommand{\Ml}{M_{\Logic}}

\newcommand{\EVd}{\Vdash^{\exists}}

\newcommand{\logicnamestyle}[1]{\mathsf{#1}}
\newcommand{\axiomnamestyle}[1]{\mathsf{#1}}

\newcommand{\CPL}{\logicnamestyle{CPL}}

\newcommand{\EM}{\logicnamestyle{M}}
\newcommand{\K}{\logicnamestyle{K}}
\newcommand{\PCLless}{\logicnamestyle{N_\less}}
\newcommand{\Nless}{\PCLless}
\newcommand{\Nlessstar}{\logicnamestyle{N_\less^*}}

\newcommand{\Logic}{\logicnamestyle{L}}
\newcommand{\logic}{\Logic}
\newcommand{\lstar}{\logicnamestyle{L^*}}
\newcommand{\logicA}{\logicnamestyle{LA}}
\newcommand{\logicU}{\logicnamestyle{LU}}

\newcommand{\NNless}{\logicnamestyle{NN_\less}}
\newcommand{\NTless}{\logicnamestyle{NT_\less}}
\newcommand{\NWless}{\logicnamestyle{NW_\less}}
\newcommand{\NCless}{\logicnamestyle{NC_\less}}
\newcommand{\NUless}{\logicnamestyle{NU_\less}}
\newcommand{\NNUless}{\logicnamestyle{NNU_\less}}
\newcommand{\NTUless}{\logicnamestyle{NTU_\less}}
\newcommand{\NWUless}{\logicnamestyle{NWU_\less}}
\newcommand{\NCUless}{\logicnamestyle{NCU_\less}}
\newcommand{\NAless}{\logicnamestyle{NA_\less}}
\newcommand{\NNAless}{\logicnamestyle{NNA_\less}}
\newcommand{\NTAless}{\logicnamestyle{NTA_\less}}
\newcommand{\NWAless}{\logicnamestyle{NWA_\less}}
\newcommand{\NCAless}{\logicnamestyle{NCA_\less}}
\newcommand{\NNlessstar}{\logicnamestyle{NN_\less^*}}
\newcommand{\NTlessstar}{\logicnamestyle{NT_\less^*}}
\newcommand{\NWlessstar}{\logicnamestyle{NW_\less^*}}
\newcommand{\NClessstar}{\logicnamestyle{NC_\less^*}}   
\newcommand{\NUlessstar}{\logicnamestyle{NU_\less^*}} 
\newcommand{\NAlessstar}{\logicnamestyle{NA_\less^*}} 

\newcommand{\Vless}{\logicnamestyle{V_\less}}
\newcommand{\VNless}{\logicnamestyle{VN_\less}}
\newcommand{\VTless}{\logicnamestyle{VT_\less}}
\newcommand{\VWless}{\logicnamestyle{VW_\less}}
\newcommand{\VCless}{\logicnamestyle{VC_\less}}
\newcommand{\VUless}{\logicnamestyle{VU_\less}}
\newcommand{\VNUless}{\logicnamestyle{VNU_\less}}
\newcommand{\VTUless}{\logicnamestyle{VTU_\less}}
\newcommand{\VWUless}{\logicnamestyle{VWU_\less}}
\newcommand{\VCUless}{\logicnamestyle{VCU_\less}}
\newcommand{\VAless}{\logicnamestyle{VA_\less}}
\newcommand{\VNAless}{\logicnamestyle{VNA_\less}}
\newcommand{\VTAless}{\logicnamestyle{VTA_\less}}
\newcommand{\VWAless}{\logicnamestyle{VWA_\less}}
\newcommand{\VCAless}{\logicnamestyle{VCA_\less}}

\newcommand{\CPR}{\axiomnamestyle{cpr}}

\newcommand{\TR}{\axiomnamestyle{tr}}
\newcommand{\axOR}{\axiomnamestyle{or}}
\newcommand{\axCO}{\axiomnamestyle{co}}
\newcommand{\axCPA}{\axiomnamestyle{cpa}}
\newcommand{\axTR}{\TR}
\newcommand{\axN}{\axiomnamestyle{n}}
\newcommand{\axT}{\axiomnamestyle{t}}
\newcommand{\axW}{\axiomnamestyle{w}}
\newcommand{\axC}{\axiomnamestyle{c}}
\newcommand{\axUone}{\axiomnamestyle{u\textup{-}}}
\newcommand{\axUtwo}{\axiomnamestyle{u}}
\newcommand{\axAone}{\axiomnamestyle{a\textup{-}}}
\newcommand{\axAtwo}{\axiomnamestyle{a}}

\newcommand{\SNlessstar}{\mathsf{G.N_\less^*}}

\newcommand{\BPCLlessstar}{\HNlessstar}

\newcommand{\HNlessstar}{\mathsf{H.N}_\less^*}
\newcommand{\HNAlessstar}{\mathsf{H.NA}_\less^*}
\newcommand{\HNless}{\mathsf{H.N}_\less}
\newcommand{\HNNless}{\mathsf{H.NN}_\less}
\newcommand{\HNTless}{\mathsf{H.NT}_\less}
\newcommand{\HNWless}{\mathsf{H.NW}_\less}
\newcommand{\HNCless}{\mathsf{H.NC}_\less}

\newcommand{\SNlessc}[1]{\mathsf{G.N#1}_\less}

\newcommand{\seqrule}[1]{$\mathsf{#1}$}
\newcommand{\leftrule}[1]{$\mathsf{{#1}_{L}}$}
\newcommand{\rightrule}[1]{$\mathsf{{#1}_{R}}$}
\newcommand{\textleftrule}[1]{\leftrule{#1}}
\newcommand{\textrightrule}[1]{\rightrule{#1}}

\newcommand{\seqruleWtwo}{\seqrule{W_0}}

\newcommand{\lctr}{\textleftrule{ctr}}
\newcommand{\rctr}{\textrightrule{ctr}}

\newcommand{\cut}{\seqrule{cut}}

\newcommand{\CO}[1]{co{#1}}
\newcommand{\COSigma}{\CO{\Sigma}}

\newcommand{\PhiBless}{\Phi_{B\less}}

\newcommand{\ih}{i.h.}
\newcommand{\ie}{i.e.}
\newcommand{\etal}{et al.}

\newcommand{\finv}[1]{\textcolor{black}{#1}}
\newcommand{\finvh}[1]{\textcolor{black}{#1}}
\newcommand{\Gc}{\G^\less}
\newcommand{\Dc}{\D^\less}
\newcommand{\Sc}{\Sigma^\less}

\newcommand{\seqrulem}[1]{\mathsf{#1}}
\newcommand{\leftrulem}[1]{\mathsf{{#1}_{L}}}
\newcommand{\rightrulem}[1]{\mathsf{{#1}_{R}}}
\newcommand{\initm}{\seqrulem{init}}

\newcommand{\ltom}{\leftrulem{\!\to}}
\newcommand{\rtom}{\rightrulem{\!\to}}

\newcommand{\lbotm}{\leftrulem{\bot}}
\newcommand{\seqruleCPNm}{\seqrulem{CP}}
\newcommand{\seqruleNm}{\seqrulem{N}}
\newcommand{\seqruleTm}{\seqrulem{T}}
\newcommand{\seqruleWm}{\seqrulem{W}}
\newcommand{\seqruleWtwom}{\seqrulem{W_0}}
\newcommand{\seqruleCm}{\seqrulem{C_0}}
\newcommand{\seqruleAm}{\seqrulem{A}}
\newcommand{\seqruleNAm}{\seqrulem{N^A}}

\newcommand{\llessm}{\leftrulem{\!\less}}
\newcommand{\rlessm}{\rightrulem{\!\less}}
\newcommand{\jumpm}{\seqrulem{jp}}

\newcommand{\manyless}{\Sigma_n}
\newcommand{\manyd}[2]{\mathbf{D}^{#2}_{#1}}
\newcommand{\manyds}{\mathbf{D}}

\newcommand{\manycdand}{\AND_{i \leq n}(C_i \less D_i) }

\newcommand{\hrulenm}{\mathsf{N}}
\newcommand{\hruletm}{\mathsf{T}}
\newcommand{\hrulewm}{\mathsf{W}}
\newcommand{\hrulecm}{\mathsf{C}}

\newcommand{\hruleaml}{\mathsf{A_L}}
\newcommand{\hruleamr}{\mathsf{A_R}}

\newcommand{\wkc}{\mathsf{wk_C}}
\newcommand{\wkb}{\mathsf{wk_B}}
\newcommand{\wkm}{\mathsf{wk}}
\newcommand{\wkrm}{\mathsf{wk_R}}
\newcommand{\wklm}{\mathsf{wk_L}}

\newcommand{\ctrc}{\mathsf{ctr_C}}
\newcommand{\ctrb}{\mathsf{ctr_B}}
\newcommand{\ctrr}{\mathsf{ctr_R}}
\newcommand{\ctrl}{\mathsf{ctr_L}}   
\newcommand{\ctrlm}{\mathsf{ctr_L}} 

\newcommand{\red}[1]{\textcolor{black}{#1}}
\newcommand{\cutm}{\mathsf{cut}}
\newcommand{\ctrtm}{\mathsf{ctr_R}}
\newcommand{\purple}[1]{\textcolor{black}{#1}}

\newcommand{\ON}{\mathcal{O}}
\newcommand{\NT}{\N^\mathsf{T}}
\newcommand{\NC}{\N^\mathsf{C}}
\newcommand{\NA}{\N^\mathsf{A}}

\newcommand{\lwkm}{\mathsf{wk_L}}
\newcommand{\rwkm}{\mathsf{wk_R}}
\newcommand{\lctrm}{\mathsf{ctr_L}}
\newcommand{\rctrm}{\mathsf{ctr_R}}

\def\lastname{Dalmonte and Girlando}

\begin{document}

\begin{frontmatter}
\title{Comparative plausibility in neighbourhood models: axiom systems and sequent calculi}
  \author{Tiziano Dalmonte$^{a}$ \quad Marianna Girlando$^{b}$}\footnote{This work was supported by the UKRI Future Leaders Fellowship `Structure vs Invariants in Proofs' MR/S035540/1,  by the SPGAS and CompRAS projects at the Free University of Bozen-Bolzano, and by the EU H2020 project INODE (grant agreement No 863410).
  } 
  \address{${}^a$Free University of Bozen-Bolzano, Bolzano, Italy \\
${}^b$University of Birmingham, Birmingham, UK}

  \begin{abstract}
We introduce a family of comparative plausibility logics over neighbourhood models, generalising Lewis' comparative plausibility operator over sphere models. We provide axiom systems for the logics, and prove their soundness and completeness with respect to the semantics. Then, we introduce two kinds of analytic proof systems for several logics in the family: a multi-premisses sequent calculus in the style of Lellmann and Pattinson, for which we prove cut admissibility, and a hypersequent calculus
	based on structured calculi for conditional logics by Girlando et al., tailored for countermodel construction over failed proof search. Our results constitute the first steps in the definition of a unified proof theoretical framework for logics equipped with a comparative plausibility operator.
  \end{abstract}

 \end{frontmatter}

\section{Introduction}
\label{sec:intro}

In the
seminal work \emph{Counterfactuals}~\cite{lewis},
besides 
the 
well-known analysis of counterfactual sentences,
David Lewis 
defined 
a 
notion of 
comparative plausibility
which has then become a standard.%
\footnote{Lewis~\cite{lewis} refers to $\less$ as the operator for \emph{comparative possibility}.
	In the literature, 
	the same or similar
	operators also go under the names of 
	\emph{entrenchment}~\cite{lellpatt}, \emph{comparative similarity}~\cite{sheremet:2005} or \emph{relative likelihood}~\cite{halpern:1997}.
	Here we adopt the terminology of, e.g.,~\cite{olivetti2015standard}.}
Specifically, Lewis 
introduced a \emph{comparative plausibility operator} $ A \less B $, read ``$ A $ is at least as plausible as $ B $'', 
which is evaluated on the plausibility ordering of worlds of a model.

Lewis' notion of comparative plausibility is defined over
\emph{sphere models}. 
These are possible-world models
in which 
every world $ x $ is endowed with a system of spheres $ S(x) $, that is, a set 
of sets of worlds
such that for every two sets in the class, one of the two is included in the other
(if $ \alpha, \beta \in S(x) $, then $\alpha\subseteq \beta$ or $\beta\subseteq\alpha$). 
This property, 
known  as \emph{nesting}, 
determines a total 
ordering 
over the set of worlds belonging to a system of spheres, 
where worlds in the inner spheres are taken to be more plausible than worlds in the outer spheres. 
Then, $ A \less B $ is true at a world $x$ if the innermost sphere in $ S(x) $ containing a world which forces $ B $ also contains a world that forces $ A $. 
The operator $\less$ is interdefinable with Lewis' \emph{conditional operator} $ A \cond B $
expressing counterfactual sentences
(formally, $ A \cond B $ is equivalent to $ (\bot \less A) \vlor \vlne ((A \vlan \vlne B) \less (A \vlan B) ) $).%

Other than in Lewis' work, several operators expressing forms of similarity or closeness between states of affairs or concepts have been 
studied in the literature, and find applications in many areas of computer science and philosophy. 
In knowledge representation,  
Sheremet \etal\ developed in \cite{sheremet:2005,sheremet2007logic} the  \emph{logic of comparative concept similarity}, evaluated over distance models, which implements a description logic-like formalism for reasoning about  
similarity of concepts in ontologies. 
Refer to \cite{alenda2009comparative} for a Lewis-style semantics for this logic. 
Moreover, similarity operators can be used in deontic reasoning to express degrees of urgency of obligations \cite{brown:1996} or, more recently, to express the preferred scenario an agent would choose in an ethical decision-making process \cite{lorini2021logic}. In philosophical logic, a logic equipped with an operator to express \emph{ceteris paribus} preference between states of affairs was introduced by Von Wright in \cite{von1972logic}, and formalised in \cite{van2009everything}. Moreover, a logic expressing \emph{ceteris paribus} preferences in a deontic setting was recently defined in \cite{loreggia2022modelling}.

A natural semantics to express generalized forms of Lewis' comparative plausibility is \emph{preferential semantics}. 
Preferential models consist of a set of worlds equipped with an explicit preorder relation $ \leq_x $ for every world $ x $, 
encoding similarity or preference among worlds. 
These models represent a generalisation of sphere models, where totality of the ordering is not assumed, and have been studied as a semantics for a family of conditional logic weaker than Lewis' counterfactual logic, called \emph{Preferential Conditional Logics} \cite{burgess:1981,halpern1994complexity},
strongly related to non-monotonic logic $ P $ from \cite{kraus:1990}.
In \cite{halpern:1997}, Halpern proposes partially ordered preferential structures as a general framework to represent forms of  preference or similarity.

We here propose a setting even more general than preferential semantics,  
by interpreting the comparative plausibility operator over \emph{neighbourhood models} (Sec.~\ref{sec:semantics}). 
These possible-worlds models are endowed with a neighbourhood function which assigns to every world $ x $
a set of sets of worlds, $ N(x) $, where 
nesting is not assumed. 
In this weaker setting, the truth condition for $ \less $ can be taken to express forms of similarity of closeness between states or concepts, which are not assumed to be totally ordered.
Neighbourhood models were introduced to define a semantics for non-normal modal logics   \cite{scott1970advice,montague1970pragmatics} and, among other applications, have been employed as a semantics for conditional logics \cite{marti,negri2015sequent,girlando:2021}.

We introduce axiom systems for the family of \emph{logics of Comparative Plausibility in Neighbourhood models} 
(\CPN\ logics)
, and prove their adequacy with respect to
some relevant classes of 
neighbourhood models
(Sec.~\ref{sec:axioms}). 
We then study the proof theory of \CPN\ logics, by defining two kinds of proof systems for 
them.
Our calculi are inspired from analytic proof systems for Lewis' conditional logics 
introduced in the literature.

We first present a \emph{multi-premisses} sequent calculus in the style of Lellmann and Pattinson \cite{lellpatt,lellmann:phd} (Sec.~\ref{sec:sequent}). The rules of these calculi display a number of premisses which depends on the number of comparative plausibility formulas occurring in the conclusion. The calculi for \CPN\ logics  represent simpler fragments of the calculi for Lewis' logics presented in \cite{lellmann:phd}. We prove cut-admissibility for the multi-premisses calculi. 
While these calculi have strong proof-theoretical properties, they are not best suited for root-first proof search: due to the fact that the comparative plausibility rules are not invertible, a heavy use of backtracking is needed to construct derivations. 

This motivates the introduction of a second family of proof systems, based on \emph{hypersequents} (Sec.~\ref{sec:hypersequents}). The calculi are inspired from the structured calculi for Lewis' logics introduced in \cite{olivetti2015standard,girlando2016}, which introduce an additional structural connective to Gentzen-style sequents representing $ \less $-formulas.
Following a strategy adopted e.g.~in \cite{dalmonte2021hypersequent} in the context of non-normal modal logics,
we further enrich the structure of sequents from \cite{girlando2016} by introducing hypersequent-style calculi, and show that they simulate the multi-premisses calculi.
Thanks to this richer structure we obtain invertibility of \emph{all} the rules in the calculus, 
which we would not have using the sequent structure from \cite{girlando2016}, 
and a more direct construction of countermodels from branches of failed proof search trees.
We conclude by discussing related works and further research directions (Sec.~\ref{sec:rel works}).

\section{Neighbourhood semantics}
\label{sec:semantics}

For $\atm = \{p_0, p_1, p_2, ...\}$ denumerable set
of propositional variables, 
we consider the formulas of $\lan$ be defined by the
BNF 
grammar
$A ::= p \mid \bot \mid A \to A \mid A \less A$,
where $p$ is any element of $\atm$,
and $\less$ is the operator for comparative plausibility. 
We assume $\top, \neg, \land, \lor$ to be defined as usual in terms of $\bot, \to$.

\begin{definition}\label{def:semantics}
	A \emph{neighbourhood model}
	is a 
	tuple
	$\M = \langle \W, \N, \V \rangle$,
	where $\W$ is a non-empty set of 
	worlds, 
	$\V$ is a valuation function $\atm \longrightarrow \pow(\W)$,
	and $\N$ is a function $\W \longrightarrow \pow(\pow(\W))$,
	called \emph{neighbourhood function},
	satisfying the 
	non-emptiness condition: for all $w\in\W$: $\emptyset\notin\N(w)$.%
	\footnote{
		Non-emptiness could be dropped as it has no impact on the 
		satisfiability
		of $\less$-formulas \cite{lewis}. 
		We assume it as it allows for a clean formulation of the conditions for the extensions, and for 
		uniformity with the neighbourhood semantics of conditional logics from \cite{negri2015sequent,girlando:2021}.
	}
	For all $w\in\W$ and $A\in\lan$,
	the forcing relation $\M, w \Vd A$ is defined inductively 
	as follows: 
	\begin{center}
		\begin{tabular}{lcl}
			$\M, w \Vd p$ & iff & $w \in \V(p)$. \\
			$\M, w \not\Vd \bot$. \\
			$\M, w \Vd B \to C$ & iff & if $\M, w \Vd B$, then $\M, w \Vd C$. \\
			$\M, w \Vd B \less C$ & iff & for all $\alpha\in\N(w)$, if there is $v\in\alpha$ s.t.~$\M, v \Vd C$, \\
			&& then there is $u\in\alpha$ s.t.~$\M, u \Vd B$.
		\end{tabular}
	\end{center}
	%
	We say that $A$ is \emph{valid in a model} $\M$, written $\M\models A$,
	if $\M, w \Vd A$ for all worlds $w$ of $\M$,
	and it is \emph{valid on a class of models} $\C$ if $\M\models A$ for all $\M\in\C$.
\end{definition}

In the following we simply write $w \Vd A$ when $\M$ is clear from the context.
We shall also use 
$\alpha\EVd A$ 
as an abbreviation for `there is $w \in \alpha$ such that $w \Vd A$'.
Thus, 
we can rewrite the
forcing clause of $\less$-formulas, graphically represented in Fig.~\ref{fig:picture}, as follows:
\begin{center}
	\begin{tabular}{lll}
		$w \Vd B\less C$ & iff & for all $\alpha\in\N(w)$, if $\alpha\EVd C$, then $\alpha\EVd B$. \\
	\end{tabular}
\end{center}

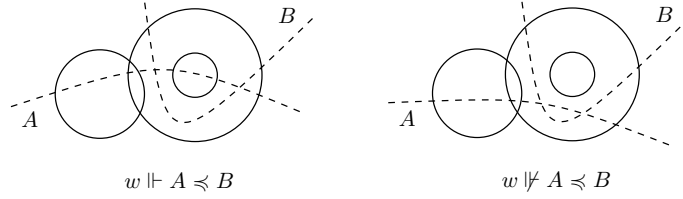
\begin{figure}[t]
	\centering
	\resizebox{26em}{!}{
		\begin{tikzpicture}		
		\draw[line width=0.2mm] (-0.5, -0.1) circle (10pt );
		\draw[line width=0.2mm, ] (-0.5, -0.1) circle (30pt );
		\draw[line width=0.2mm, ] (-2, -0.4) circle (20pt );
		\node (B) at (0.95,0.85) {$B$};
		\node (A) at (-3.1,-0.8) {$A$};
		\node (w) at (-0.75,-1.8) {$w \Vd A \less B$};
		\draw[dashed] (1.35,0.8) [line width=0.2mm, align=center] .. controls (-0.95,-1.45) ..  (-1.3,1.1);		
		\draw[dashed] (-3.4,-0.6) [line width=0.2mm, align=center] .. controls (-0.95,0.2) ..  (1.2,-0.7);
		\end{tikzpicture}
		\qquad\quad
		\begin{tikzpicture}		
		\draw[line width=0.2mm] (-0.5, -0.1) circle (10pt );
		\draw[line width=0.2mm, ] (-0.5, -0.1) circle (30pt );
		\draw[line width=0.2mm, ] (-2, -0.4) circle (20pt );
		\node (B) at (0.95,0.85) {$B$};
		\node (A) at (-3.1,-0.8) {$A$};
		\node (w) at (-0.75,-1.8) {$w \not\Vd A \less B$};
		\draw[dashed] (1.35,0.8) [line width=0.2mm, align=center] .. controls (-0.95,-1.45) ..  (-1.3,1.1);		
		\draw[dashed] (-3.4,-0.55) [line width=0.2mm, align=center] .. controls (-0.95,-0.45) ..  (1.1,-1.25);
		\end{tikzpicture}
	}
	\caption{\label{fig:picture} Representation of the forcing clause for $\less$-formulas.
		Dashed lines represent the extensions of $A$ and $B$.}
\end{figure}

We observe that unary modalities can be defined on the basis of $\less$, namely 
$\Box_K A := \bot\less\neg A$, and  $\Box_M A := \neg(\neg A \less \top)$,
where $\Box_K$ and $\Box_M$ are the Box modalities of respectively logic $\K$ and 
non-normal logic $\EM$ (cf.~e.g~\cite{pacuit}).
%
Note also that 
sphere models can be recovered by adding the condition of \emph{nesting} to the \neigh\ function: for all $ \alpha, \beta \in \N(w) $, either $ \alpha \subseteq \beta$ or $ \beta \subseteq \alpha $. 
Lewis considered in \cite{lewis} several additional properties, which turn out to be of interest when formulated on the \neigh\ function. We consider here classes of \neigh\ models satisfying combinations of the following properties:
\begin{center}
	\begin{tabular}{c ll}
		$ \mathsf{N}  $& $\N(w)\not=\emptyset$. \hfill (\emph{Normality}) \\
		$ \mathsf{T} $&There is $\alpha\in\N(w)$ such that $w \in \alpha$. \hfill (\emph{Total reflexivity})  \\
		$ \mathsf{W} $& $\N(w)\not=\emptyset$ and for all $\alpha\in\N(w)$, $w \in \alpha$. \hfill (\emph{Weak centering}) \\
		$ \mathsf{C} $& $\{w\} \in \N(w)$ and for all $\alpha\in\N(w)$, $w \in \alpha$. \hfill (\emph{Centering}) \\
		$ \mathsf{U} $& If $v \in\alpha$ and $\alpha\in\N(w)$, then $\bigcup\N(v) = \bigcup\N(w)$. \qquad \quad (\emph{Uniformity}) \\
		$ \mathsf{A} $&If $v \in\alpha$ and $\alpha\in\N(w)$, then $\N(v) = \N(w)$. \hfill (\emph{Absoluteness}) \\ 
	\end{tabular}
\end{center}
The condition of absoluteness can also be formulated as follows:
\begin{center}
	\begin{tabular}{c ll}
		$ \mathsf{A+} $& For all $ v,w \in \W $,  $\N(v) = \N(w)$. \hspace{1.8cm} \hfill (\emph{Strong Absoluteness}) \\ 
	\end{tabular}
\end{center}
Equivalence of  $\mathsf{A} $ and $  \mathsf{A+} $ over formulas validity can be easily established using the same strategy described by Lewis \cite[p.~122]{lewis}. We will use $  \mathsf{A+} $ in Sec.~\ref{sec:hypersequents}.

Neighbourhood semantics can be used to express a variety of situations. By means of example, 
let $\alpha, \beta, \gamma, ...$ in $\N(w)$ represent sources of information available at $w$ which are 
not arranged in any priority or reliability order. 
In this setting, $A \less B$ expresses that $A$ is at least as plausible as $B$ in that
whenever $w$ receives information $B$, it also receives information $A$.
Then, model conditions represent natural assumptions about the information sources:
by \emph{normality},  every 
$w$ has a 
source of information available,
while according to \emph{reflexivity} or \emph{weak centering},
$w$ belongs to some or all of the sources available to itself
(e.g.~online forums $w$ must be registered at).
Moreover, \emph{uniformity} and \emph{absoluteness} express kinds of information bubbles,
since  if $v$ belongs to a source 
available to $w$,
then $w$ and $v$ have access to the same sources of information.

\section{Axiom systems for \CPN\ logics}
\label{sec:axioms}

In this section we present the logics of Comparative Plausibility in Neighbourhood models
(\CPN\ logics in the following)
corresponding to 
the classes of neighbourhood models introduced in Sec.~\ref{sec:semantics}.
We propose axiom systems for \CPN\ logics and show their soundness and completeness. 
Then, we compare \CPN\ logics with Lewis' logics of comparative plausibility in sphere models.

\begin{definition}\label{def:systems}
	\CPN\ logics are defined 
	by extending classical 
	propositional logic ($\CPL$) formulated in $\lan$ with the 
	rules and axioms for $\less$ from Fig.~\ref{fig:axioms}:
	\begin{center}
		\begin{tabular}{llllll}
			$\Nless := \CPL \cup\{\axTR, \axOR, \CPR\}$ &
			$\NNless := \Nless \cup \{\axN\}$  &
			$\NTless := \Nless \cup \{\axT\}$ \\
			$\NWless := \NTless \cup \{\axW\}$ &
			$\NCless := \NWless \cup \{\axC\}$ \\
		\end{tabular}
	\end{center}
	Moreover, 
	for $ \logic \in \{\Nless, \NNless, \NTless,\NWless,\NCless \} $, 
	we define
	$\logicU_\less := \logic \cup \{\axUone, \axUtwo\}$ and
	$\logicA_\less := \logic \cup \{\axAone, \axAtwo\}$. 
\end{definition}

\begin{figure}
	\fbox{\begin{minipage}{32.5em}
			\centering
			\begin{small}
				\begin{tabular}{lllllll}
					& & & & & & \\[-0.5cm]
					\multicolumn{2}{l}{\multirow{ 2}{*}{$\CPR$ \ax{$A \to B$}\uinf{$B \less A$}\disp}} &&
					$\TR$ & $(A \less B) \land (B \less C) \to (A \less C)$ \\
					
					\vspace{0.1cm}
					&&& $\axOR$ & $(A \less B) \land (A \less C) \to (A \less B \lor C)$ \\ 
					$\axN$ & $\neg(\bot\less\top)$ && $\axUone$ & $\neg(\bot\less A) \to (\bot\less(\bot\less A))$ \\
					$\axT$ & $(\bot\less A) \to \neg A$  && $\axUtwo$ & $(\bot\less A) \to (\bot\less\neg(\bot\less A))$ \\
					$\axW$ & $A \to (A \less \top)$ && $\axAone$ & $(A\less B) \to (\bot\less\neg(A\less B))$ \\
					\vspace{0.1cm}
					$\axC$ & $(A\less\top) \to A$ && $\axAtwo$ & $\neg(A\less B) \to (\bot\less(A\less B))$ \\[-0.3cm]
				\end{tabular}
			\end{small}
	\end{minipage}}
	\caption{\label{fig:axioms} Axioms and rules for \CPN\ logics.} 
\end{figure}

The logics generated by this definition are displayed in the lower layer of
the lattice of 
systems in Fig.~\ref{fig:dyagram}.
The axioms of $ \Nless $ are those defined by Lewis in \cite[Ch.6]{lewis}, while axioms for extensions of $ \Nless $ are  reformulations of Lewis' axioms in terms of $\less$ \cite{lellpatt,girlando2016}. 
In the following, 
for every logic $\logic$ from Def.~\ref{def:systems}, we denote $\lstar$ any extension of $\logic$.
As usual, we say that a formula $A$ is 
\emph{derivable} in $\Nlessstar$, written 
$\vd_\Nlessstar A$, 
if there is a finite sequence of formulas ending with $A$ where every formula is an axiom of $\Nlessstar$,
or is obtained from previous formulas by \emph{modus ponens} or $\CPR$.
Moreover, we say that $A$ is \emph{deducible} in $\Nlessstar$ from a set of formulas $\Phi$ if there is a finite set 
$\{B_1, ..., B_n\}\subseteq\Phi$ such that 
$\vd_\Nlessstar  B_1 \land ... \land B_n \to A$.

For each logic $\Nlessstar$, we call \emph{$\Nlessstar$-model} any neighbourhood model
satisfying the conditions corresponding to the letters 
appearing beside $\mathsf{N}$
in the name of the logic. Thus, $\Nless$-models denotes the class of all \neigh\ models,   $\NNless$-models the class of all 
models satisfying normality, and so on.

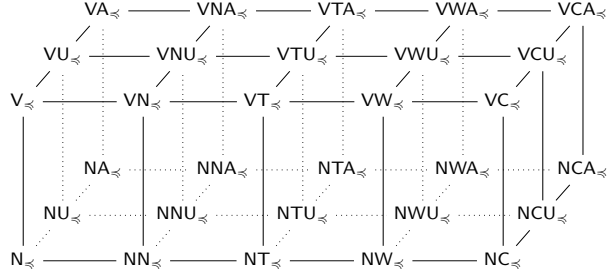
\begin{figure}[t]
	\centering
	\resizebox{8.4cm}{3.7cm}{
		\begin{small}
			\begin{tikzpicture}
			\node (n) at  (0,0)  {$\Nless$};
			\node (nn) at (2, 0) {$\NNless$};
			\node (nt) at (4, 0) {$\NTless$};
			\node (nw) at (6, 0) {$\NWless$};
			\node (nc) at (8, 0) {$\NCless$};
			
			\node (nu) at  (0.66,0.8)  {$\NUless$};
			\node (nnu) at (2.66,0.8) {$\NNUless$};
			\node (ntu) at (4.66,0.8) {$\NTUless$};
			\node (nwu) at (6.66,0.8) {$\NWUless$};
			\node (ncu) at (8.66,0.8) {$\NCUless$};
			
			\node (na) at  (1.33,1.6)  {$\NAless$};
			\node (nna) at (3.33,1.6) {$\NNAless$};
			\node (nta) at (5.33,1.6) {$\NTAless$};
			\node (nwa) at (7.33,1.6) {$\NWAless$};
			\node (nca) at (9.33,1.6) {$\NCAless$};
			
			\node (v) at  (0,2.8)  {$\Vless$};
			\node (vn) at (2, 2.8) {$\VNless$};
			\node (vt) at (4, 2.8) {$\VTless$};
			\node (vw) at (6, 2.8) {$\VWless$};
			\node (vc) at (8, 2.8) {$\VCless$};
			
			\node (vu) at  (0.66,3.6)  {$\VUless$};
			\node (vnu) at (2.66,3.6) {$\VNUless$};
			\node (vtu) at (4.66,3.6) {$\VTUless$};
			\node (vwu) at (6.66,3.6) {$\VWUless$};
			\node (vcu) at (8.66,3.6) {$\VCUless$};
			
			\node (va) at  (1.33,4.4)  {$\VAless$};
			\node (vna) at (3.33,4.4) {$\VNAless$};
			\node (vta) at (5.33,4.4) {$\VTAless$};
			\node (vwa) at (7.33,4.4) {$\VWAless$};
			\node (vca) at (9.33,4.4) {$\VCAless$};
			
			\draw (n) -- (nn);
			\draw (nn) -- (nt);
			\draw (nt) -- (nw);
			\draw (nw) -- (nc);
			
			\draw [dotted] (nu) -- (nnu);
			\draw [dotted] (nnu) -- (ntu);
			\draw [dotted] (ntu) -- (nwu);
			\draw [dotted] (nwu) -- (ncu);
			
			\draw [dotted] (na) -- (nna);
			\draw [dotted] (nna) -- (nta);
			\draw [dotted] (nta) -- (nwa);
			\draw [dotted] (nwa) -- (nca);
			
			\draw [dotted] (n) -- (nu);
			\draw [dotted] (nn) -- (nnu);
			\draw [dotted] (nt) -- (ntu);
			\draw [dotted] (nw) -- (nwu);
			\draw (nc) -- (ncu);
			
			\draw [dotted] (nu) -- (na);
			\draw [dotted] (nnu) -- (nna);
			\draw [dotted] (ntu) -- (nta);
			\draw [dotted] (nwu) -- (nwa);
			\draw (ncu) -- (nca);
			
			\draw (v) -- (vn);
			\draw (vn) -- (vt);
			\draw (vt) -- (vw);
			\draw (vw) -- (vc);
			
			\draw (vu) -- (vnu);
			\draw (vnu) -- (vtu);
			\draw (vtu) -- (vwu);
			\draw (vwu) -- (vcu);
			
			\draw (va) -- (vna);
			\draw (vna) -- (vta);
			\draw (vta) -- (vwa);
			\draw (vwa) -- (vca);
			
			\draw (v) -- (vu);
			\draw (vn) -- (vnu);
			\draw (vt) -- (vtu);
			\draw (vw) -- (vwu);
			\draw (vc) -- (vcu);
			
			\draw (vu) -- (va);
			\draw (vnu) -- (vna);
			\draw (vtu) -- (vta);
			\draw (vwu) -- (vwa);
			\draw (vcu) -- (vca);
			
			\draw (n) -- (v);
			\draw (nn) -- (vn);
			\draw (nt) -- (vt);
			\draw (nw) -- (vw);
			\draw (nc) -- (vc);
			
			\draw [dotted] (nu) -- (vu);
			\draw [dotted] (nnu) -- (vnu);
			\draw [dotted] (ntu) -- (vtu);
			\draw [dotted] (nwu) -- (vwu);
			\draw (ncu) -- (vcu);
			
			\draw [dotted] (na) -- (va);
			\draw [dotted] (nna) -- (vna);
			\draw [dotted] (nta) -- (vta);
			\draw [dotted] (nwa) -- (vwa);
			\draw (nca) -- (vca);
			\end{tikzpicture}
		\end{small}
	}
	\caption{\label{fig:dyagram} The family of \CPN\ logics. The system in the upper layer are Lewis' logics.}
\end{figure}

We show that each logic $\Nlessstar$ is characterised by the class of all $\Nlessstar$-models. 
We first prove 
that the logics are sound with respect to the corresponding classes of models.

\begin{theorem}[Soundness]\label{th:soundness}
	For every formula $A$, if $A$ is derivable in $\Nlessstar$, then $A$ is valid in all $\Nlessstar$-models.
\end{theorem}
\begin{proof}
	We show that the modal axioms and rules of $\Nlessstar$ are valid (resp.~sound) in the corresponding models,
	considering some relevant examples.
	($\CPR$) Assume $\M \models A \to B$,
	and $\alpha\in\N(w)$, $\alpha\EVd A$. Then $\alpha\EVd B$,
	therefore $w \Vd B \less A$.
	($\axTR$) If $w \Vd (A \less B) \land (B \less C)$, then for every $\alpha\in\N(w)$, 
	$\alpha\EVd B$ implies $\alpha\EVd A$, and $\alpha\EVd C$ implies $\alpha\EVd B$,
	then $\alpha\EVd C$ implies $\alpha\EVd A$,
	therefore $w \Vd A \less C$.
	($\axOR$) If $w \Vd (A \less B) \land (A \less C)$, then for every $\alpha\in\N(w)$, 
	$\alpha\EVd B$ implies $\alpha\EVd A$, and $\alpha\EVd C$ implies $\alpha\EVd A$,
	then $\alpha\EVd B \lor C$ implies $\alpha\EVd A$,
	therefore $w \Vd A \less B \lor C$.
	($\axN$) By normality, for all $w\in\W$ there is $\alpha\in\N(w)$. Moreover, $\alpha\not\EVd\bot$ and since $\alpha\not=\emptyset$, $\alpha\EVd\top$, then $w\Vd\neg(\bot\less\top)$.
	($\axT$) Assume $w\Vd\bot\less A$. Then for all $\alpha\in\N(w)$, $\alpha\not\EVd A$. Moreover by total reflexivity, 
	there is $\beta\in\N(w)$ such that $w\in\beta$. Then $w\not\Vd A$, thus $w\Vd\neg A$.
	($\axW$) Assume $w\Vd A$. By weak centering, $\N(w)\not=\emptyset$, and $w\in\alpha$ for all $\alpha\in\N(w)$.
	Then for all $\alpha\in\N(w)$, $\alpha\EVd A$,  thus $w\Vd \top\less A$.
	($\axC$) Assume $w\Vd A\less \top$. Then for all $\alpha\in\N(w)$, $\alpha\EVd A$.
	Moreover by centering, $\{w\}\in\N(w)$, therefore $w\Vd A$.
\end{proof}

Using a canonical model construction inspired from \cite{lewis}, we shall now prove that $\Nlessstar$ is complete with respect to the class of all $\Nlessstar$-models.
As usual, for any logic $\Logic$ and set of formulas $\Phi$,
we say that $\Phi$ is $\Logic$-\emph{consistent}
if 
$\Phi\not\vd_{\Logic} \bot$,
and that it is $\Logic$-\emph{maximal consistent} (maxcons) if it is consistent and
for every 
$B \notin\Phi$, 
$\Phi \cup \{B\} \vd_{\Logic} \bot$.
The proof of the following Lemma is standard.

\begin{lemma}\label{lemma:lindenbaum}
	(a) If $\Phi$ is a $\Logic$-consistent set of formulas, then there is a $\Logic$-maximal consistent set $\Psi$
	such that $\Phi\subseteq\Psi$.
	(b) If $\Phi$ is a $\Logic$-maximal consistent set,
	then for all $A, B\in\lan$, 
	(i) if $\Phi\vd_{\Logic} A$, then $A\in\Phi$;
	(ii) $A \in\Phi$ if and only if $\neg A \notin\Phi$;
	(iii) if $A \lor B \in\Phi$, then $A\in\Phi$ or $B\in\Phi$.
\end{lemma}

We consider the following notion of \emph{cut around},%
\footnote{This terminology comes from Lewis~\cite{lewis}, but our definition 
	is different from Lewis' one.
}
and prove the subsequent lemma that will be needed in the following.

\begin{definition}\label{def:cut around}
	Let $\Phi$ be a maximal consistent set of formulas, and $\Sigma$ be a set of formulas.
	We say that $\Sigma$ is a \emph{cut around} 
	$\Phi$ if for all finite sets $\{B_1, ..., B_n\} \subseteq \Sigma$ and all $A\notin\Sigma$,
	$(B_1 \lor ... \lor B_n) \less A \notin \Phi$. 
	Moreover,  
	let $\COSigma = \{\Psi \textup{ maxcons} \mid \Psi\cap\Sigma=\emptyset\}$.
\end{definition}

\begin{lemma}\label{lemma:cosphere} 
	If $\Sigma$ is a cut around $\Phi$ for some maximal consistent set $\Phi$, then for every formula $A$,
	$A \in \Sigma$ if and only if for all $\Psi\in\COSigma$, $A\notin\Psi$.
\end{lemma}
\begin{proof}
	If $A \in\Sigma$ and $\Psi\in\COSigma$, then $\Psi\cap\Sigma = \emptyset$,
	thus $A \notin \Psi$.
	If instead $A \notin\Sigma$, then suppose by contradiction that
	$\{\neg B \mid B \in \Sigma\} \cup \{A\} \vd_{\Logic}\bot$.
	Then there are formulas $B_1, ..., B_n \in \Sigma$ such that 
	$\vd_{\Logic} \neg B_1 \land ... \land \neg B_n \to \neg A$,
	thus
	$\vd_{\Logic} A \to B_1 \lor ... \lor B_n$, therefore 
	$\vd_{\Logic} (B_1 \lor ... \lor B_n) \less A$.
	By closure under derivation, $(B_1 \lor ... \lor B_n) \less A \in \Phi$,
	but by definition of cut around, $(B_1 \lor ... \lor B_n) \less A \notin \Phi$.
	We conclude that $\{\neg B \mid B \in \Sigma\} \cup \{A\} \not\vd_{\Logic}\bot$.
	Then by Lemma~\ref{lemma:lindenbaum}, there is $\Psi\in\Wl$ such that
	$\{\neg B \mid B \in \Sigma\} \cup \{A\} \subseteq \Psi$,
	therefore $\Psi\in\COSigma$ and $A \in \Psi$.
\end{proof}

From Lemma~\ref{lemma:cosphere} it immediately follows that $\bot\in\Sigma$ for all $\Sigma$ cut around $\Phi$.
We 
now define the canonical model.

\begin{definition}\label{def:canonical model}
	For every \CPN\ logic $\logic$, 
	the \emph{canonical model} for $\Logic$ is the tuple $\Ml = \langle \Wl, \Nl, \Vl \rangle$,
	where:
	
	\begin{itemize}
		\item $\Wl$ is the class of all $\Logic$-maximal consistent sets;
		\item for all $\Phi \in \Wl$, $\Nl(\Phi) = \{\COSigma \mid \Sigma \textup{ cut around } \Phi \textup{ and } \COSigma \not=\emptyset\}$;
		\item $\Vl(p) = \{\Phi \in \Wl \mid p \in \Phi\}$.
	\end{itemize}
\end{definition}

\begin{lemma}[Truth lemma]\label{lemma:truth}
	If $\logic$ is a \CPN\ logic and $\Ml$ is the canonical model for $\logic$, then 
	for all $A\in\lan$ and all $\Phi\in\Wl$, $\Phi\Vd A$ if and only if $A \in \Phi$.
\end{lemma}
\begin{proof}
	By induction on the construction of $A$. 
	For atomic and propositional formulas the proof is standard.
	We consider the case $A = B \less C$.
	
	($\Rightarrow$) 
	Suppose $\Phi\Vd B \less C$.
	Then for all $\alpha\in\Nl(\Phi)$, $\alpha\EVd C$ implies $\alpha\EVd B$.
	By definition, this means that for every $\Sigma$ cut around $\Phi$, 
	$\COSigma\EVd C$ implies $\COSigma\EVd B$.
	Let $\PhiBless = \{D \mid B \less D \in \Phi\}$. 
	Then since $B \less B \in\Phi$, $B \in \PhiBless$.
	Moreover, $\PhiBless$ is a cut around $\Phi$:
	if $E_1, ..., E_n \in \PhiBless$ and $F \notin\PhiBless$,
	then $B \less E_1, ..., B \less E_n \in \Phi$, and $B \less F \notin\Phi$.
	Thus by axiom $\axOR$ and closure under derivation,
	$B \less (E_1 \lor ... \lor E_n) \in \Phi$,
	whence by $\TR$, $(E_1 \lor ... \lor E_n) \less F\notin\Phi$.
	Now suppose by contradiction that $\{\neg D \mid D \in \PhiBless\} \cup \{C\} \not\vd_{\Logic} \bot$.
	Then by Lemma~\ref{lemma:lindenbaum} there is $\Psi\in\Wl$ such that
	$\{\neg D \mid D \in \PhiBless\} \cup \{C\} \subseteq\Psi$.
	We then have $C \in \Psi$ and $\PhiBless \cap \Psi = \emptyset$, 
	which implies $\Psi\in\CO{\PhiBless}$.
	By \ih, $\Psi \Vd C$,
	thus $\CO{\PhiBless} \EVd C$, which implies $\CO{\PhiBless} \EVd B$.
	This means that there is $\Omega \in \CO{\PhiBless}$ such that $\Omega \Vd B$,
	therefore by \ih, $B \in \Omega$.
	Furthermore, by definition we have $\Omega \cap \PhiBless = \emptyset$,
	then $B \notin \PhiBless$, which contradicts $B \in \PhiBless$.
	Therefore $\{\neg D \mid D \in \PhiBless\} \cup \{C\} \vd_{\Logic} \bot$.
	Then 
	there are
	$D_1, ..., D_n \in \PhiBless$ 
	such that
	$\vd_{\Logic} \neg D_1 \land ... \land \neg D_n \to \neg C$,
	that is  $\vd_{\Logic} C \to D_1 \lor ... \lor D_n$,
	whence by $\CPR$, 
	$(D_1 \lor ... \lor D_n) \less C \in \Phi$.
	Moreover by definition of $\PhiBless$, 
	$B \less D_1, ..., B \less D_n \in \Phi$.
	Then by $\axOR$, $B \less (D_1 \lor ... \lor D_n) \in \Phi$, 
	finally by $\TR$, $B \less C \in \Phi$.

	$(\Leftarrow)$
	Suppose $\Phi\not\Vd B \less C$.
	Then there is $\alpha\in\Nl(\Phi)$ such that $\alpha\EVd C$ and $\alpha\not\EVd B$,
	i.e., there is a $\Sigma$ cut around $\Phi$ with $\COSigma\EVd C$ and $\COSigma\not\EVd B$.
	By \ih\ there is $\Psi\in\COSigma$ such that $C\in\Psi$, and 
	for all $\Omega\in\COSigma$, $B\notin\Omega$.
	Then $C \notin\Sigma$, and from Lemma~\ref{lemma:cosphere} it follows that
	$B\in\Sigma$. Then by definition $B \less C \notin\Phi$.
\end{proof}

\begin{lemma}[Model lemma]\label{lemma:model}
	The canonical model for $\Nlessstar$ 
	is a 
	$\Nlessstar$-model.
\end{lemma}
\begin{proof}
	Non-emptiness is immediate. 
	We consider the other conditions. 
	
	($\NNlessstar$) 
	For every 
	$\Phi\in\W_{\NNless}$, $\neg(\bot\less\top)\in\Phi$, then by Lemma~\ref{lemma:truth}, 
	$\Phi\Vd \neg(\bot\less\top)$,
	thus there is $\alpha\in\N_{\NNless}(\Phi)$ such that $\alpha\EVd \top$ and $\alpha\not\EVd \bot$.
	
	($\NTlessstar$) 
	For any $\Phi\in\W_{\NTless}$, let $\Sigma=\{A \mid \bot\less A\in\Phi\}$. 
	$\Sigma$ is a cut around $\Phi$, since for all $B_1, ..., B_n \in \Sigma$ and $C \notin\Sigma$,
	$\bot\less B_1, ..., \bot\less B_n \in\Phi$ and $\bot\less C\notin\Phi$, 
	then by $\axOR$, $\bot\less B_1 \lor ... \lor B_n\in\Phi$,
	and by $\axTR$, $B_1 \lor ... \lor B_n\less C\notin\Phi$.
	Moreover, for any $A\in\lan$, 
	if $A \in\Sigma$, then $\bot\less A\in\Phi$, thus by $\axT$, $\neg A\in\Phi$, 
	whence $A\notin\Phi$.
	Thus $\Sigma\cap\Phi= \emptyset$, which implies $\Phi\in\COSigma$,
	and since $\COSigma\not= \emptyset$, $\COSigma\in\N_{\NTless}(\Phi)$.
	
	($\NWlessstar$) Since $\NWless\vd\axN$, by item ($\NNless$),  $\N_{\NWless}(\Phi)\not=\emptyset$ for all $\Phi\in\W_{\NWless}$.
	Moreover let $\COSigma\in\N_{\NWless}(\Phi)$ for a $\Sigma$ cut around $\Phi$.
	Then there is $\Psi\in\COSigma$.
	Since $\top\in\Psi$, by Lemma~\ref{lemma:cosphere}, $\top\notin\Sigma$,
	then for all $A \in\Sigma$, $A \less \top \notin\Phi$, thus by axiom $\axW$, $A\notin\Phi$.
	This means $\Sigma\cap\Phi=\emptyset$, therefore $\Phi\in\COSigma$.

	($\NClessstar$) Since axiom $\axW$ belongs to $\NCless$, by item ($\NWless$), 
	for all $\Phi\in\W_{\NCless}$ and all $\alpha\in\N_{\NCless}(\Phi)$, $\Phi\in\alpha$.
	Moreover, let $\Sigma = \{A \mid A\less\top \notin\Phi\}$.
	Then $\Sigma$ is a cut around $\Phi$: 
	if $B_1, ..., B_n\in\Sigma$ and $C\notin\Sigma$, then $B_1 \less \top, ..., B_n \less\top\notin\Phi$ and $C\less\top\in\Phi$.
	By axiom $\axW$, $B_1\not\in\Phi$, ..., $B_n\notin\Phi$,
	thus $B_1\lor ... \lor B_n\notin\Phi$,
	then by axiom $\axC$, $(B_1\lor ... \lor B_n)\less\top\notin\Phi$,
	therefore by $\axTR$, $(B_1\lor ... \lor B_n)\less C\notin\Phi$.
	Moreover, since $\top\less\top\in\Phi$, $\top\notin\Sigma$,
	by Lemma~\ref{lemma:cosphere}
	there is $\Psi\in\COSigma$,
	thus $\COSigma\in\N_{\NCless}(\Phi)$.
	Suppose $\Psi\not=\Phi$.
	Then there is $A\in\lan$ such that $A\in\Psi$ and $A\notin\Phi$.
	Since $\Psi\in\COSigma$, $A \notin\Sigma$, 
	then $A\less\top\in\Phi$,
	thus by $\axC$, $A\in\Phi$, it follows $\Psi=\Phi$, therefore $\COSigma=\{\Phi\}$.
	
	($\NUlessstar$)
	Suppose $\Psi\in\bigcup\N_{\NUlessstar}(\Phi)$.
	Then $\Psi\in\COSigma$ for some $\Sigma$ cut around $\Phi$.
	We show that for all $A\in\lan$, $\bot\less A\in\Phi$ iff $\bot\less A\in\Psi$.
	If $\bot\less A\in\Phi$, then by axiom $\axUtwo$, $\bot\less\neg(\bot\less A)\in\Phi$,
	then by Def.~\ref{def:cut around},  $\bot\notin\Sigma$ or $\neg(\bot\less A)\in\Sigma$.
	Since $\bot\in\Sigma$, we have $\neg(\bot\less A)\in\Sigma$, thus $\neg(\bot\less A)\notin\Psi$, then $\bot\less A\in\Psi$.
	If $\bot\less A\in\Psi$, then $\bot\less A\notin\Sigma$, thus $\bot\less(\bot\less A)\notin\Phi$,
	then by axiom $\axUone$, $\neg(\bot\less A)\notin\Phi$, therefore $\bot\less A\in\Phi$.
	Let $\Pi = \{A \mid \bot\less A\in\Phi\} = \{A \mid \bot\less A\in\Psi\}$.
	Then $\Pi$ is a cut around $\Phi$ and $\Psi$:
	If $B_1, ..., B_n\in\Pi$ and $C\notin\Pi$, then
	$\bot\less B_1, ..., \bot\less B_n\in\Phi$ and $\bot\less C\notin\Phi$, 
	thus by $\axOR$, $\bot\less B_1 \lor ... \lor B_n\in\Phi$,
	then by $\axTR$, $B_1 \lor ... \lor B_n\less C\notin\Phi$.
	Moreover for all $\Omega$ cut around $\Phi$ or $\Psi$, $\Pi\subseteq\Omega$,
	therefore $\CO\Omega\subseteq\CO\Pi$.
	Then in particular $\CO\Sigma\subseteq\CO\Pi$, which implies $\Psi\in\CO\Pi$, 
	thus $\CO\Pi\not=\emptyset$.
	It follows $\CO\Pi \in\N_{\NUlessstar}(\Phi)$ and $\CO\Pi \in \N_{\NUlessstar}(\Psi)$,
	therefore
	$\CO\Pi = \bigcup\N_{\NUlessstar}(\Phi) = \bigcup\N_{\NUlessstar}(\Psi)$.
	
	($\NAlessstar$)
	Suppose $\Psi\in\bigcup\N_{\NAlessstar}(\Phi)$.
	Then $\Psi\in\COSigma$ for some $\Sigma$ cut around $\Phi$.
	We show that for all $A,B\in\lan$, $A\less B\in\Phi$ iff $A\less B\in\Psi$.
	If $A\less B\in\Phi$, then by axiom $\axAone$, $\bot\less\neg(A\less B)\in\Phi$,
	then by Def.~\ref{def:cut around},  $\bot\notin\Sigma$ or $\neg(A\less B)\in\Sigma$.
	Since $\bot\in\Sigma$, we have $\neg(A\less B)\in\Sigma$, thus $\neg(A\less B)\notin\Psi$, then $A\less B\in\Psi$.
	If $A\less B\in\Psi$, then $A\less B\notin\Sigma$, thus $\bot\less(A \less B)\notin\Phi$,
	then by axiom $\axAtwo$, $\neg(A\less B)\notin\Phi$, therefore $A\less B\in\Phi$.
	It follows that for every $\Pi$, $\Pi$ is a cut around $\Phi$ iff $\Pi$ is a cut around $\Psi$,
	thus $\CO{\Pi}\in\N_{\NAlessstar}(\Phi)$ iff $\CO{\Pi}\in\N_{\NAlessstar}(\Psi)$,
	therefore $\N_{\NAlessstar}(\Phi) = \N_{\NAlessstar}(\Psi)$.
\end{proof}

As a consequence of the previous lemmas 
we obtain the following result.

\begin{theorem}[Completeness]
	\label{thm:completeness_models_axioms}
	For every formula $A$, if $A$ is valid in all $\Nlessstar$-models, then $A$ is derivable in $\Nlessstar$.
\end{theorem}
\begin{proof}
	Suppose $\not\vd_\Nlessstar A$. Then $\{\neg A\}$ is $\Nlessstar$-consistent, 
	thus by Lemma~\ref{lemma:lindenbaum} there is a 
	$\Nlessstar$-maxcons
	set $\Phi$ such that
	$\neg A \in \Phi$.
	By definition, $\Phi \in\W_{\Nlessstar}$, and by Lemma~\ref{lemma:truth}, 
	$\M_{\Nlessstar}, \Phi\not\Vd A$,
	moreover by Lemma~\ref{lemma:model}, $\M_{\Nlessstar}$ is a $\Nlessstar$-model.
\end{proof}


\vspace{0.2cm}

Let us now turn to the relationship between 
$\Nlessstar$ 
and Lewis' logics 
of comparative plausibility over 
sphere models. 
Lewis \cite{lewis} provides two equivalent axiomatisations of the minimal 
logic  $\Vless$,
one of the two being
$\CPL \cup \{\CPR, \axCPA, \axTR, \axCO\}$, where $\axCO$ is the \emph{connection axiom} $(A \less B) \lor (B \less A)$,
and $\axCPA$ is  $(A \less A \lor B) \lor (B \less A \lor B)$.
We show that a further equivalent axiomatisation
of $\Vless$ 
can be given by extending
our minimal logic $\Nless$ with the axiom $\axCO$:

\begin{proposition}
	For all $A\in\lan$, 
	$\vd_{\Vless} A$
	if and only if 
	$\vd_{\Nless \cup \{\axCO\}} A$.
\end{proposition}
\begin{proof}
	Since $\Nless \cup \{\axCO\}$ and $\Vless = \CPL \,\cup\, \{\CPR, \axCPA, \axTR, \axCO\}$ differ only with respect to 
	$\axCPA$ and $\axOR$, it suffices to show that
	(i) $\vd_{\Nless \cup \{\axCO\}}\axCPA$ and (ii) $\vd_\Vless \axOR$.
	(i) From $(A \less B) \lor (B \less A)$, by $\CPR$ we have
	$((A \less A) \land (A \less B)) \lor ((B \less A) \land (B \less B))$,
	then by $\axOR$, $(A \less A \lor B) \lor (B \less A \lor B)$.
	(ii) From $(B \less B \lor C) \lor (C \less B \lor C)$, by $\axTR$ we have
	$(A \less B) \land (A \less C) \to (A \less  B \lor C) \lor (A \less B \lor C)$,
	thus $(A \less B) \land (A \less C) \to (A \less  B \lor C)$.
\end{proof}

Note also that the extensions of $\Nless$ are defined by the same axioms characterising the extensions of $\Vless$.
It follows that each Lewis' logic can be 
obtained from
the corresponding CPN logic
by adding the connection axiom $\axCO$.
The relations among 
these systems are displayed in Fig.~\ref{fig:dyagram}.

\section{Multi-premisses sequent calculi for \CPN\ logics}
\label{sec:sequent}

In this section we present Gentzen-style sequent calculi for the 
\CPN\ logics $ \Nless, \NNless, \NTless,\NWless, \NCless, \NAless $, and $ \NNAless $. 
From now on, let $\Nlessstar$ denote any of these systems.
For each logic we introduce a calculus $\SNlessstar$ 
defined on the basis of the sequent systems for Lewis' logics 
by Lellmann and Pattinson \cite{lellmann:phd,lellpatt}. In these calculi, the rules have up to $ m+n $ premisses, where $ m $ (resp.~$ n $) is the number of $\less$-formulas occurring in the antecedent (resp.~consequent) of the conclusion. 
Calculi for 
$\Nlessstar$ can be provided by restricting the calculi in \cite{lellmann:phd,lellpatt} to rules with at most one $ \less $-formula in the consequent ($ n = 1 $), thus obtaining simpler calculi, where each rule introduces at most $ m+1  $ premisses.  

As usual, we call \emph{sequent}
any pair $\G \seq \D$, 
where $\G$ and $\D$ are finite, possibly empty multisets of formulas of $\lan$. 
$ \G \seq \D $ is interpreted as the formula $\AND\G \to \OR\D$. 

The rules of the calculi $\SNlessstar$ can be found in  Fig.~\ref{fig:sequent rules}.
Each modal rule simultaneously analyses a number (at least one)
of $\less$-formulas appearing in a sequent. 
The \emph{principal formulas} of each  rule $ \mathsf{R}_n $ are the $ n $ or $ n+1 $ $ \less $-formulas in the conclusion which get analysed in the premiss.
In some rules the principal $\less$-formulas are copied into the premisses 
in order to ensure admissibility of contraction. 
We denote derivability in $\SNlessstar$ as $ \SNlessstar \vd \G\seq\D$.

\begin{figure}[t]
	\begin{adjustbox}{max width = \textwidth}
		\begin{tabular}{|@{\hspace{-0.4cm}} c @{\hspace{-0.4cm}}|}
			
			\hline 
			~  \\[0.01cm]

			$ \vlinf{\initm}{}{\G, p\seq p, \D}{} $
			\qquad
			$ \vlinf{\lbotm}{}{\G, \bot\seq  \D}{} $
			\qquad 
			$ \vliinf{\ltom}{}{\G, A \to B \seq \D}{\G \seq A, \D}{\G, B \seq \D} $
			\qquad 
			$ \vlinf{\rtom}{}{\G \seq A \to B, \D}{\G, A \seq B, \D} $\\[0.4cm]

			$ \vliinf{\seqruleCPNm_n}{}{\G, C_1 \less D_1, \dots, C_n \less D_n \seq A \less B, \D}
			{ \{C_k \seq A, D_1, \dots , D_{k-1}\}_{1\leq k\leq n} }{B \seq A, D_1, \dots , D_n} $\\[0.4cm]

			$ \vliinf{\seqruleNm_n}{}{\G, C_1 \less D_1, \dots, C_n \less D_n \seq \D}{\{C_k \seq D_1, \dots, D_{k-1}\}_{1 \leq k\leq n}}{ \seq D_1, \dots , D_n} $\\[0.4cm]

			$ \vliinf{\seqruleTm_n}{}{\G, C_1 \less D_1, \dots, C_n \less D_n \seq \D}{\{C_k \seq D_1, \dots, D_{k-1}\}_{1 \leq k\leq n}}{\G, \finv{C_1 \less D_1, \dots, C_n \less D_n} \seq D_1, \dots, D_n, \D} $\\[0.4cm]

			$ \vliinf{\seqruleWm_n}{}{\G, C_1 \less D_1, \dots , C_n \less D_n \seq A \less B, \D}{\{C_k \seq A, D_1, \dots, D_{k-1}\}_{1 \leq k\leq n}}{\G , \finv{C_1 \less D_1, \dots , C_n \less D_n} \seq A, \finv{A \less B}, D_1, \dots, D_n, \D} $\\[0.4cm]

			$ \vlinf{\seqruleWtwom}{}{\G \seq A \less B, \D}{\G \seq \finv{A \less B}, A, \D} $
			\qquad 
			$ \vliinf{\seqruleCm}{}{\G, A \less B \seq \D}{\G, \finv{A \less B},A \seq \D}{\G, \finv{A \less B} \seq B, \D} $\\[0.4cm]
			
			$ \vliinf{\seqruleAm_n}{}{\G, C_1 \less D_1,\dots, C_n \less D_n \seq A \less B, \D }
			{ \{ \Gc, \finv{\Sc}, C_k \seq \finv{A \less B}, A, D_1, \dots, D_{k-1}, \Dc\}_{1\leq k\leq n} }
			{\Gc, \finv{\Sc}, B \seq \finv{A \less B}, A, D_1, \dots, D_n, \Dc} $
			\\[0.5cm]
			$ \vliinf{\seqruleNAm_n}{}{\G, C_1 \less D_1, \dots, C_n \less D_n \seq \D }
			{ \{ \Gc, \finv{\Sc}, C_k \seq  D_1, \dots, D_{k-1}, \Dc\}_{1 \leq k\leq n} }
			{\Gc, \finv{\Sc} \seq  D_1, \dots, D_n, \Dc} $\\[0.5cm]
			\multicolumn{1}{|l|}{\small where \, $ \finv{\Sc} = \finv{C_1 \less D_1, ..., C_n \less D_n}$, \, $\Gc = \{C' \less D' \mid C' \less D' \in \Gamma\} $, \, $ \Dc = \{C' \less D' \mid C' \less D' \in \Delta\} $ }\\
			\hline 
			\multicolumn{1}{c}{}\\[-0.2cm]
			\multicolumn{1}{c}{
				\begin{tabular}{l l}
					\multicolumn{2}{c}{
						$\SNlessc{} = \{\initm, \lbotm, \ltom, \rtom\} \cup \{\seqruleCPNm_n \mid n \geq 0 \}$ }\\
					$\SNlessc{N} = \SNlessc{} \cup \{\seqruleNm_n \mid n \geq 1 \}$ &
					$\SNlessc{T} = \SNlessc{} \cup \{\seqruleTm_n \mid n \geq 1 \}$ \\
					$\SNlessc{W} = \SNlessc{} \cup \{\seqruleWm_n \mid n \geq 0 \} \cup \{ \seqruleTm_n\mid n \geq 1 \}$ & 
					$\SNlessc{C} = \SNlessc{} \cup \{\seqruleWtwom, \seqruleCm\} $  \\
					$\SNlessc{A} = \{\initm, \lbotm, \ltom, \rtom\} \cup \{\seqruleAm_n \mid n \geq 0 \}$
					&
					$\SNlessc{NA} = \SNlessc{A} \cup \{\seqruleNAm_n \mid n \geq 1 \}$
					\\
					
				\end{tabular}
			}
		\end{tabular}
	\end{adjustbox}
	\caption{\label{fig:sequent rules} Gentzen-style calculi for 
		\CPN\ logics. }
\end{figure}

\begin{theorem}[Soundness]
	\label{th:soundness SNless}
	If $\SNlessstar \vd \G \seq \D$  then $\Nlessstar\vd\AND\G \to \OR\D$.
\end{theorem}

\begin{proof}
	We  show that 
	for every rule $R$ of $\SNlessstar$
	with premisses $\G_1 \seq \D_1$, \dots, $\G_n\seq\D_n$ 
	and conclusion $\G\seq\D$,
	the corresponding Hilbert-style rule with premisses 
	$\AND\G_1 \to \OR\D_1$, \dots, $\AND\G_n \to \OR\D_n$ 
	and conclusion $\AND\G \to \OR\D$ is derivable in $\Nlessstar$.
	The propositional cases are standard. 
	We use $ \manyd{1}{n} $ as a shorthand for $ D_1, \dots, D_n $.

	($\seqruleCPNm_n$)
	Suppose $\vd C_1 \to A$, $\vd C_2 \to A \lor D_1$, \dots, $\vd C_n \to A \vlor  \bigvee \manyd{1}{n-1} $ and $ \vd B \to A \lor \bigvee \manyd{1}{n} $.
	Then by $\CPR$, we have $\vd A \less C_1$, $\vd A \lor D_1 \less C_2$,  \dots , $  \vd A \lor \bigvee \manyd{1}{n-1} \less C_n $ and $\vd A \lor \bigvee \manyd{1}{n} \less B$.
	If $ n = 0 $, the conclusion immediately follows. 
	From $\vd C_1 \less D_1$ it follows by $ \axTR $ that $ A \less D_1 $. 
	Thus, $\vd (C_1 \less D_1) \to (A \less D_1)$. 
	Since $\vd A \less A$, 
	by $\axOR$, we have that $\vd (C_1 \less D_1) \to (A \less A \lor D_1)$.
	By $\axTR$, $\vd (C_1 \less D_1) \to (A \less C_2)$. 
	From $\vd C_2 \less D_2$ it follows by $ \axTR $ that $ A \less C_2 $. 	
	Thus, $\vd (C_1 \less D_1) \land (C_2 \less D_2) \to (A \less D_2)$.
	By $ \axOR $ applied to $ \vd A \less A$ and to $ A \less D_1 $, we have 
	$\vd (C_1 \less D_1) \land (C_2 \less D_2) \to (A \less A \vlor D_1 \vlor D_2)$, to which we apply $ \axTR $ twice and conclude 
	$\vd (C_1 \less D_1) \land (C_2 \less D_2) \land (C_3 \less D_3)\to (A \less D_3)$. 
	We iterate the steps above until we obtain $\vd \manycdand  \to (A \less D_n)$. 
	Then, by applications of $\axOR$ to $ A \less A $ and to $ A\less D_1 ,\dots, A \less D_{n-1} $, we obtain 
	$\vd \manycdand \to (A \less A \lor \bigvee \manyd{1}{n})$. 
	A final application of $ \axTR $ 
	yields 
	$\vd \manycdand \to (A \less B)$. 
	Therefore $\vd \AND\G \land  \manycdand \to (A \less B) \lor \OR\D$ for every $\G, \D$.

	($\seqruleNm_n$)
	Suppose $\vd C_1 \to \bot$, $\vd C_2 \to D_1$, \dots, $\vd C_n \to \OR\manyd{1}{n-1}$ and $\vd \top\to \OR \manyd{1}{n}$.
	Then by $\CPR$, $\vd \bot \less C_1$, $\vd D_1 \less C_2$, \dots, $\vd \OR \manyd{1}{n-1} \less C_n$ and $\vd \OR \manyd{1}{n}\less \top$.
	Reasoning as in the case of $ \seqruleCPNm_n $, we conclude that 
	$\vd \manycdand \to (\bot  \less \top)$. 
	By $\axN$, $\vd \manycdand \to \bot$,
	then $\vd \AND\G \land \manycdand \to \OR\D$ for all $\G, \D$.

	($\seqruleTm_n$)
	Suppose $\vd C_1 \to \bot$, $\vd C_2 \to D_1$, \dots, $\vd C_n \to \OR \manyd{1}{n-1}$ and $\vd \AND\G \land  \manycdand \to \OR \manyd{1}{n}  \lor \OR\D$.
	Then by $\CPR$, $\vd \bot \less C_1$,  $\vd D_1 \less C_2$, \dots $ \vd \OR \manyd{1}{n-1} \less C_n $.
	By applications of $\axTR$ and $ \axOR $, we have that $\vd \manycdand \to (\bot \less D_1) \land \dots \land (\bot \less D_n)$. 
	By $\axT$,   $\vd \manycdand \to \neg D_1 \land \dots \land \neg D_n$.
	Then we have $\vd \AND\G  \land  \manycdand  \to (\OR \manyd{1}{n}\OR\D) \land \neg D_1 \land \dots \land \neg D_n$,
	from which we conclude 
	that $\vd \AND\G \land \manycdand \to \OR\D$.

	($\seqruleWm_n$)
	Suppose $\vd C_1 \to A$, $\vd C_2 \to A \lor D_1$, \dots, $ \vd C_n \to A \lor \OR \manyd{1}{n-1} $ and $\vd \AND\G \land \manycdand  \to A  \lor \OR \manyd{1}{n} \lor A \less B \lor \OR\D$.
	Then by $\CPR$, $\vd A \less C_1$, $ \vd A \lor  D_1 \less C_2 $, \dots, $ \vd A, \OR \manyd{1}{n-1} $. Reasoning as in the case of $ \seqruleCPNm_n $, we obtain proofs of the following: 
	$ \vd (C_1 \less D_1) \to (A \less D_1) $, \dots, $ \vd \manycdand \to (A \less D_n)$. 
	Moreover by $\axW$, 
	$\vd \AND\G \land \manycdand  \to (A \less \top) \vlor \OR_{i\leq n} (D_i \less \top) \lor A \less B \lor \OR\D$.
	Thus, applying $ \axTR $ to $ A \less D_1 $, \dots , $A \less D_n  $, we obtain 
	$\vd \AND\G \land\manycdand \to (A \less \top)  \lor A \less B \lor\OR\D$.
	Since by $\CPR$, $\vd \top\less B$ for every $B$,
	by $\axTR$ we obtain
	$\vd \AND\G \land \manycdand \to (A \less B)  \lor A \less B \lor\OR\D$.
	
	(\seqruleWtwo)
	If $\vd \AND\G \to A \lor \OR\D$, then by $\axW$, $\vd \AND\G \to (A\less\top) \lor \OR\D$,
	thus since by $\CPR$, $\vd\top\less B$ for every $B$, we have
	$\vd \AND\G \to (A\less B) \lor \OR\D$.
	
	($\seqruleCm$)
	Suppose $\vd \AND\G \land A \to \OR\D$ and $\vd \AND\G \to B \lor \OR\D$.
	Then 
	by $\axW$, $\vd \AND\G \to (B\less\top) \lor \OR\D$,
	thus by $\axTR$, $\vd \AND\G \land (A \less B) \to (A\less\top) \lor \OR\D$.
	Then by $\axC$, $\vd \AND\G \land (A \less B) \to A \lor \OR\D$,
	therefore $\vd \AND\G \land (A \less B) \to \OR\D$.

	($\seqruleAm_n$) 
	Suppose $\vd \AND\Gc \land \manycdand \vlan C_1 \to A \vlor (A \less B) \lor \OR \Dc$, 
	$\vd \AND\Gc \land \manycdand \vlan C_2 \to A \vlor (A \less B) \lor D_1 \vlor \OR \Dc$, \dots, $\vd \AND\Gc \land \manycdand \vlan C_n \to A \vlor  \bigvee \manyd{1}{n-1} \vlor \OR \Dc $ and 
	$ \vd \AND\Gc \land \manycdand \vlan B \to A \lor \bigvee \manyd{1}{n}  \vlor \OR \Dc$.
	Using the same strategy as in $ \seqruleCPNm_n $, we prove that	$\vd \manycdand \to (A \less B)$ follows from the simpler set of assumptions where we remove $  \AND\Gc $, $  \manycdand $, $ \OR \Dc $  and $ A \less B $. 
	From this, we conclude $\vd  \AND\G \land \manycdand \manycdand \to (A \less B)  \lor \OR \D$ for any $ \G $, $ \D $.
	($\seqruleNAm_n$)  is similar.	
\end{proof}

We now show that $\SNlessstar$ enjoy 
cut admissibility,
where a rule is said to be \emph{(height-preserving) admissible} if, whenever the premisses are derivable, also the conclusion is derivable (with a derivation of at most the same height).
We start by considering the following auxiliary result. 
\begin{proposition}
	\label{prop:seq_adm_wk_ctr}
	The rules below
	are height-preserving
	admissible in $\SNlessstar$:
	\vspace{-0.1cm}
	\begin{center}
		$
		\vlinf{\lwkm}{}{ \G, A \seq \D}{\G \seq \D}
		\quad 
		\vlinf{\rwkm}{}{\G \seq A, \D}{\G \seq \D}
		\quad 
		\vlinf{\lctrm}{}{\G, A \seq \D}{\G, A, A \seq \D}
		\quad 
		\vlinf{\rctrm}{}{\G \seq A, \D}{\G \seq A, A, \D}
		$
	\end{center}
\end{proposition}
\begin{proof}
	By induction on the height $ h $ of the derivation of the premiss of the rules. The cases of $ \wklm $, $ \wkrm $ and \rctr\ are immediate. We show admissibility of \lctr, for $ h>0 $ and the last rule applied in the derivation being $ \seqruleCPNm_{n+1} $.  Let $ C_i = C_{i+1} $ and $ D_i = D_{i+1} $. 
	\begin{center}
		\begin{small}
			$
			\vlderivation{
				\vlin{\ctrlm}{}{
					G, C_1 \less D_1, \dots, C_i \less D_i, C_{i+2} \less D_{i+2},\dots, C_n \less D_n \seq A \less B, \D
				}{
					\vliin{\seqruleCPNm_{n}}{}{G, C_1 \less D_1, \dots, C_i \less D_i, C_{i+1} \less D_{i+1} ,\dots, C_n \less D_n \seq A \less B, \D}{
						\vlhy{\{C_k \seq A, D_1, \dots, D_{k-1}\}_{ k\leq n}}
					}{ 
						\vlhy{B \seq A, D_1, \dots, D_i, D_{i+1} , \dots, D_n }
					}
				}
			}
			$
		\end{small}
	\end{center}
	For each $ j \geq i +2 $ apply contraction to the following sequent of smaller height:
	\begin{center}
		\begin{small}
			$
			\vlderivation{
				\vlin{\ctrlm}{}{C_{j} \seq A, D_1, \dots, D_i,D_{i+2},  \dots, D_{j-1} }{\vlhy{C_{j} \seq A, D_1, \dots,  D_i,D_{i+1}, \dots, D_{j-1} }}
			}
			$
		\end{small}
	\end{center}
	A final application of $ \ctrtm $ on smaller height and $ \seqruleCPNm_{n-1} $ yields the desired result:
	
	\ 
	
	\begin{adjustbox}{max width = \textwidth}
		$
		\vlderivation{
			\vliin{\seqruleCPNm_{n-1}}{}{G, C_1 \less D_1, \dots, C_i \less D_i, C_{i+2} \less D_{i+2} ,\dots, C_n \less D_n \seq A \less B, \D}{
				\vlhy{
					\begin{matrix}
					\{C_k \seq A, D_1, \dots, D_{k-1}\}_{ k\leq i}\\
					\{C_{j} \seq A, D_1, \dots, D_i,D_{i+2} \dots, D_{j-1}\}_{i+2\leq j \leq n}
					\end{matrix}
				}
			}{ 
				\vlin{\ctrlm}{}{B \seq A, D_1, \dots, D_i , \dots, D_n}{
					\vlhy{B \seq A, D_1, \dots, D_i, D_{i+1} , \dots, D_n }
				}
			}
		}
		$
	\end{adjustbox}
	
\end{proof}

\begin{theorem} 
	\label{thm:cut}
	The 
	cut rule is admissible in $\SNlessstar$, where $ A $ is the \emph{cut formula}:
	\begin{center}
		\vspace{-0.1cm}
		$
		\vliinf{\cutm}{}{\G, \G' \seq \D, \D'}{\G \seq \red A, \D}{\G', \red  A \seq \D'} 
		$
		
	\end{center}
\end{theorem}

\begin{proof}
	By induction on lexicographically ordered pairs ($c$, $h$),
	where $c$ is the complexity of the cut formula
	(\ie, the number of binary
	connectives or modalities occurring in it),
	and 
	$ h $ 
	is the sum of the heights of the derivations of the premisses of \cut. 
	We distinguish 
	cases according to whether the cut formula is
	principal
	in the last rules applied in the derivation of the premisses of \cut. 
	
	If the cut formula is not principal in the last rule application of the derivation of one of the two premisses of \cut, then the conclusion of \cut\ is standardly obtained by \ih\ on $ h $. 
	Suppose the cut formula is principal in the last rule application of the derivations of both premisses of \cut. 
	
	\
	
	\noindent $ \bullet $ \, Both premisses of \cut\ are derived by $ \seqruleCPNm_n $:
	
	\ 
	
	\begin{adjustbox}{max width = \textwidth}
		$ 
		\vlderivation{
			\vliin{\cutm}{}{\G, \G', A_1 \less B_1, \dots, A_{i-1} \less B_{i-1}, C_1 \less D_1, \dots, C_n \less D_n, A_{i+1} \less B_{i+1}, \dots, A_m \less B_m \seq E \less F, \D, \D'}{
				\vliin{\seqruleCPNm_n}{}{\G, C_1 \less D_1, \dots , C_n \less D_n \seq \red{A_i \less B_i}, \D}{
					\vlhy{\{C_k \seq A_i, D_1, \dots, D_{k-1}\}_{k\leq n}}
				}{
					\vlhy{B_i \seq A_i, D_1, \dots, D_n}
				}
			}{
				\vliin{\seqruleCPNm_n}{}{\G', A_1 \less B_1, ..., \red{A_i \less B_i}, ..., A_m \less B_m \seq E \less F, \D'}{
					\vlhy{\{A_\ell \seq E, B_1, \dots, B_{\ell-1}\}_{\ell\leq m}}
				}{
					\vlhy{F \seq E, B_1, \dots, B_i ,\dots ,  B_m}
				}
			}
		}
		$
	\end{adjustbox}
	
	\ 
	
	\noindent The derivation is converted as follows: First, for every $k \leq n$ we obtain the following derivation, by induction on $ c $:
	\vspace{-0.2cm}
	$$
	\small{
		\vliinf{\cutm}{}{C_k \seq E, B_1, \dots, B_{i-1}, D_1, \dots , D_{k-1}}{C_k \seq \red{A_i}, D_1, \dots , D_{k-1}}{\red{A_i} \seq E, B_1,\dots ,  B_{i-1}}
	}
	$$
	Always by induction on $ c $, for every $1 \leq \ell \leq n-i$ we obtain the following, where the double line denotes several applications of \rctr:
	%
	%
	%
	\begin{center}
		\begin{footnotesize}
			$ 
			\vlderivation{
				\vliq{\ctrtm}{}{ A_{i+l} \seq E, B_1, \dots, B_{i-1},D_1, \dots, D_n, B_{i+1}, \dots, B_{i+l-1} }{
					\vliin{\cutm}{}{
						A_{i+l} \seq E, E, B_1, \dots, B_{i-1}, B_1, \dots, B_{i-1}, D_1, \dots, D_n , B_{i+1}, \dots, B_{i+l-1}
					}{
						\vlhy{A_{i+l} \seq E, B_1, \dots, \purple{B_i}, \dots, B_{i+l-1}}
					}{
						\vliin{\cutm}{ }{ \purple{B_i} \seq E, D_1, \dots, D_n  , B_1, \dots, B_{i-1}}{ 
							\vlhy{B_i \seq \red{A_i}, D_1, \dots, D_n}
						}{
							\vlhy{\red{A_i} \seq E, B_1, \dots, B_{i-1}}
						}
					}
				}
			}
			$
		\end{footnotesize}
		%
		%
	\end{center}
	By induction on $ c $ we construct the following derivation $\mathcal S$: 
	\begin{center}
		$
		\footnotesize{
			\vlderivation{
				\vliq{\ctrtm}{}{ F \seq E , B_1, \dots , B_{i-1}, D_1, \dots, D_n , B_{i+1}, \dots, B_m  }{
					\vliin{\cutm}{}{ F \seq E , E, B_1, \dots , B_{i-1}, B_1, \dots , B_{i-1}, D_1, \dots, D_n , B_{i+1}, \dots, B_m }{
						\vlhy{F \seq E, B_1, \dots, \purple{B_i} ,\dots ,  B_m}
					}{
						\vliin{\cutm}{ }{ \purple{B_i} \seq E, B_1, \dots, B_{i-1}, D_1, \dots, D_n  }{ 
							\vlhy{B_i \seq \red{A_i}, D_1, \dots, D_n}
						}{
							\vlhy{\red{A_i} \seq E, B_1, \dots, B_{i-1}}
						}
					}
				}
			}
		}
		$
		%
	\end{center}
	A final application of 
	$ \seqruleCPNm_{n-1} $ 
	yields a derivation of the conclusion of $ \cutm $:

	\ 
	
	\begin{adjustbox}{max width = \textwidth}
		$ 
		\vlderivation{
			\vliiin{\seqruleCPNm_{n-1}}{}{
				\G, \G', A_1 \less B_1, ..., A_{i-1} \less B_{i-1}, C_1 \less D_1, ..., C_n \less D_n, A_{i+1} \less B_{i+1}, ..., A_m \less B_m \seq E \less F, \D, \D'
			}{
				\vlhy{
					\begin{matrix}
					\{A_\ell \seq E, B_1, ..., B_{\ell-1}\}_{\ell<i}\\
					\{C_k \seq E, B_1, ..., B_{i-1}, D_1, ..., D_{k-1}\}_{k\leq n}\\
					\{A_{i+l} \seq E, B_1, ..., B_{i-1}, D_1, ..., D_n, B_{i+1}, ..., B_{i+l-1}\}_{1 \leq \ell \leq n-i}\\
					\end{matrix}
				}
			}{\vlhy{\qquad }}{
				\vlhy{
					\begin{matrix}
					~\\ 
					~\\
					\mathcal S
					\\
					\end{matrix}
				}
			}
		}
		$
	\end{adjustbox}
	
	\
	
	\noindent $ \bullet $ \, The cut formula is principal in $ \seqruleWm_n $ and $ \seqruleCPNm_m $:
	
	\ 
	
	\begin{adjustbox}{max width = \textwidth}
		$
		\vlderivation{
			\vliin{\cutm}{}{\G , \G', A_1 \less B_1, \dots, A_{i-1}\less B_{i-1}, C_1\less D_1, \dots, C_n\less D_n, A_{i+1}\less B_{i+1},\dots,  A_m \less B_m  \seq \D, \D', E \less F}{
				\vlin{\seqruleWm_m}{}{\G, C_1\less D_1, \dots, C_n\less D_n \seq \D, \red{A_i \less B_i} }
				{\vlhy{
						\begin{matrix}
						\{C_k \seq A_i, D_1, \dots D_{k-1}\}_{1 < k \leq n}\\
						\G, C_1\less D_1, \dots, C_n\less D_n \seq \D, A_i \less B_i, A_i, D_1, \dots, D_n  
						\end{matrix}
					}
				}
			}{
				\vliin{\seqruleCPNm_m}{}{\G', A_1 \less B_1, ..., \red{A_i \less B_i}, ..., A_m \less B_m \seq  \D', E \less F}{
					\vlhy{\{A_\ell \seq E, B_1, \dots, B_{\ell-1}\}_{\ell\leq m}}
				}{
					\vlhy{F \seq E, B_1, \dots, B_i ,\dots ,  B_m}
				}
			}
		}
		$
	\end{adjustbox}
	
	\
	
	\noindent We first perform a \cut\ on smaller $ h $ between the premiss of $ \seqruleWm_m $ and the rightmost premiss of \cut, obtaining sequent
	$		\Sigma  = \G, \G' ,  A_1 \less B_1, \dots, A_{i-1}\less B_{i-1}, \purple{A_i}, C_1\less D_1, \dots, C_n\less D_n , A_{i+1}\less B_{i+1},\dots, A_m \less B_m \seq \D, \D', E \less F,  D_1, \dots, D_n$. 
	
	Next, we perform a \cut\ by induction on $ c $ on $ \Sigma $ and on $ \purple{A_i} \seq E, B_1, \dots, B_{i-1} $. 
	This yields a derivation of the sequent
	$	\Sigma'   = \G, \G' ,  A_1 \less B_1, \dots, A_{i-1}\less B_{i-1}, C_1\less D_1, \dots, C_n\less D_n ,
	A_{i+1}\less B_{i+1},\dots, A_m \less B_m \seq \D, \D', E \less F,  D_1, \dots, D_n, E, B_1, \dots, B_{i-1}$
	Then, for $ 1 < k \leq n $, we generate the following derivation, by induction on $ c $:
	$$
	\vlderivation{
		\vliin{\cutm}{}{
			C_k \seq E, B_1, \dots, B_{i-1}, D_1, \dots, D_{k-1}
		}
		{\vlhy{C_k \seq \red{A_i}, D_1, \dots D_{k-1} } }
		{\vlhy{ \red{A_i} \seq E, B_1, \dots, B_{i-1} }}
	}
	$$
	A final application of $ \seqruleWm_{j} $, where $ j = n + (i-1) $, yields the desired conclusion:
	
	\
	
	\begin{adjustbox}{max width = \textwidth}
		$
		\vlderivation{
			\vliiin{\seqruleWm_j}{}{
				\G, \G' ,  A_1 \less B_1, \dots, A_{i-1}\less B_{i-1}, C_1\less D_1, \dots, C_n\less D_n, A_{i+1}\less B_{i+1},\dots, A_m \less B_m \seq \D, \D', E \less F
			}{
				\vlhy{
					\begin{matrix}
					\{A_\ell \seq E, B_1, \dots, B_{\ell-1}\}_{\ell\leq i}\\
					\{C_k \seq E, B_1, \dots, B_{i-1}, D_1, \dots, D_{k-1}\}_{k \leq n}\\
					\end{matrix}
				}
			}
			{\vlhy{\qquad \qquad }}
			{\vlhy{
					\begin{matrix}
					~\\
					\Sigma'\\
					\end{matrix}	
				}
			}
		}
		$
	\end{adjustbox}
	
	\

	\noindent 	$ \bullet $ \, The cut formula is principal in $ \seqruleWtwom$ and $ \seqruleCPNm_m $:

	\
	
	\begin{adjustbox}{max width = \textwidth}
		$
		\vlderivation{
			\vliin{\cutm}{}{\G , \G', A_1 \less B_1, \dots , A_{i-1} \less B_{i-1},  A_{i+1} \less B_{i+1}, \dots , A_m \less B_m \seq \D, \D', E \less F}{
				\vlin{\seqruleWtwom}{}{\G \seq \D, \red{A_i \less B_i} }{
					\vlhy{\G \seq \D, A_i \less B_i , A_i }
				}
			}{
				\vliin{\seqruleCPNm_m}{}{\G', A_1 \less B_1, \dots, \red{A_i \less B_i}, \dots A_m \less B_m \seq \D', E \less F}{
					\vlhy{\{A_\ell \seq E, B_1, \dots, B_{\ell-1}\}_{\ell\leq m}}
				}{
					\vlhy{F \seq E, B_1, \dots, B_i ,\dots ,  B_m}
				}
			}
		}
		$
	\end{adjustbox}
	
	\ 
	
	\noindent We first perform a \cut\ by induction on $ h $ on the premiss of $ \seqruleWtwom $ and on the rightmost premiss of \cut, obtaining sequent
	$	\Sigma  =  \G, \G', A_1 \less B_1, \dots, A_{i-1} \less B_{i-1}, A_{i+1} \less B_{i+1}, \dots A_m \less B_m \seq
	\D, \D',E \less F, \red{A_i} $. 
	Then, applying $ \cutm $ to $ \Sigma $ and $ \red{A_i }\seq E, B_1, \dots, B_{i-1} $ we obtain 
	$	\Sigma'  =  \G, \G', A_1 \less B_1, \dots, A_{i-1} \less B_{i-1}, A_{i+1} \less B_{i+1}, \dots A_m \less B_m
	\seq \D, \D',E \less F,  E, B_1, \dots, B_{i-1}$. 
	Let $ \G^* $ denote the premiss of $ \Sigma' $ and $ \D^* =  \D, \D'$. 
	We now construct the following derivation, containing $ i-1 $ applications of  $ \seqruleCm $:

	\begin{adjustbox}{max width = \textwidth}
		$ 
		\vlderivation{
			\vlin{\seqruleWtwom}{}{\G^* \seq \D^*, E \less F  }{
				\vliin{\seqruleCm}{}{
					\G^* \seq \D^*, E \less F,  E
				}{
					\vliq{\wkm}{}{A_1 , \G^* \seq \D^*, E \less F, E }{
						\vlhy{
							\begin{matrix}
							~\\
							A_1 \seq E\\
							\end{matrix}
						}
					}
				}{
					\vliin{\seqruleCm}{}{
						\begin{matrix}
						\G^* \seq \D^*, E \less F, E, B_1,\dots, B_{i-3}\\
						\vdots\\
						\G^* \seq \D^*, E \less F, E,  B_1\\	
						\end{matrix}
					}{
						\vliq{\wkm}{}{
							A_{i-2}, \G^* \seq \D^*, E \less F, E, B_1,\dots, B_{i-3}
						}{
							\vlhy{A_{i-2} \seq E, B_1,\dots, B_{i-2} }
						}
					}{
						\vliin{\seqruleCm}{}{
							\G^* \seq \D^* , E, B_1, \dots, B_{i-2}
						}{
							\vliq{\wkm}{}{A_{i-1}, \G^* \seq \D^* , E \less F, E, B_1, \dots, B_{i-2}}{
								\vlhy{A_{i-1} \seq E, B_1, \dots, B_{i-2}}
							}	
						}{
							\vlhy{\Sigma'}
						}
					}
				}
			}
		}
		$
	\end{adjustbox}
	
	\ 
	
	The remaining cases are: $ \seqruleCPNm_n $ + $ \seqruleWm_m$, which is proved similarly as  $ \seqruleWm_m$ + $ \seqruleCPNm_n $; $ \seqruleCm$ + $ \seqruleCPNm_n $, similar to $ \seqruleCPNm_n$ + $ \seqruleWtwom $; $ \seqruleCm$ + $ \seqruleWtwom $, which is immediate, 	
	$ \seqruleCPNm_n  + \seqruleNm_n $, which is proven in the same way as  $ \seqruleCPNm_n  + \seqruleCPNm_m $, $ \seqruleCPNm_n  + \seqruleTm_m $, which is proven as $ \seqruleCPNm_n  + \seqruleWm_m $, and the cases for absoluteness, which are proven as their counterpart without absoluteness.
\end{proof}

Thanks to cut-admissibility, we obtain cut-free completeness of the calculi, by deriving the axioms and inference rules of $\Nlessstar$ in $\SNlessstar$. 

\begin{corollary}[Completeness]
	\label{th:completeness SNless}
	If  $\Nlessstar\vd\AND\G \to \OR\D$ then $\SNlessstar \vd \G \seq \D$.
\end{corollary}
\begin{proof}
	Derivations of the $ \Nless $ axioms in $ \SNlessc{} $ are displayed in Fig.~\ref{fig:derivations in GN}.
	The derivations employ standard propositional rules
	for $\land$ and $\lor$, which can be defined in  $\SNlessstar$.
	The derivations of the axioms for extensions are straightforward. \emph{Modus ponens} is simulated using \cut\ in the usual way.
\end{proof}

\begin{figure}
	\footnotesize
	{\begin{tabular}{c c}
			\multicolumn{2}{c}{
				$ \vlderivation{
					\vlin{\rtom}{}{\seq (A \less B) \land (A \less C) \to (A \less B \lor C)}{
						\vlin{\vlan_\mathsf{L}}{}{(A \less B) \land (A \less C) \seq A \less B \lor C}{
							\vliiin{\seqruleCPNm_{2}}{}{A \less B, A \less C \seq A \less B \lor C }{
								\vlhy{A \seq A}}{\vlhy{A \seq A, B}}{
								\vliin{\vlor_\mathsf{L}}{}{B \lor C \seq A, B, C}{\vlhy{B \seq A, B, C}}{\vlhy{C \seq A, B, C}}
							}
						}
					}
				} $}\\
			\\
			$ \vlderivation{
				\vlin{\rtom}{}{\seq (A \less B) \land (B \less C) \to (A \less C)}{
					\vlin{\vlan_\mathsf{L}}{}{(A \less B) \land (B \less C) \seq A \less C}{
						\vliiin{\seqruleCPNm_{2}}{}{A \less B, B \less C \seq A \less C}{\vlhy{A \seq A}}{\vlhy{B, A \seq B}}{\vlhy{C \seq A, B, C}}
					}
				}
			} $
			&
			$	\vlderivation{
				\vlin{\seqruleCPNm_{0}}{}{\seq B \less A}{
					\vliin{\cutm}{}{A \seq B}{\vlhy{\seq A \to B}}{
						\vliin{\ltom}{}{ A, A \to B \seq B}{\vlhy{A \seq A, B}}{\vlhy{A, B \seq B}}
					}
				}
			}$
			\\
		\end{tabular}
	}
	\normalsize
	\caption{\label{fig:derivations in GN} Derivations of the axioms and rules of $ \Nless $  in $ \SNlessc{} $.}
\end{figure}

Termination of root-first proof search in $ \SNlessstar $ can be easily proved by observing that non redundant rule applications strictly decrease the complexity of formulas. 
However, $ \SNlessstar $ are not suited for root-first proof search: the comparative plausibility rules are \emph{not invertible}, meaning that derivability of the conclusion does not imply derivability of the premiss(es) of the rule. As a consequence, backtrack points are generated when constructing root-first a derivation. 
Next section introduces proof systems having only invertible rules.

\section{Hypersequent calculi for \CPN\ logics}
\label{sec:hypersequents}

In this section we present hypersequent calculi $\BPCLlessstar$ for the same family of \CPN\ logics treated in Sec.~\ref{sec:sequent}, 
namely $ \Nless, \NNless, \NTless,\NWless, \NCless, \NAless $ and $ \NNAless $, always denoted by 
$\Nlessstar$. 
Disregarding the hypersequent structure, the calculi $\BPCLlessstar$ are fragments of the sequent calculi for Lewis' logics
by Olivetti and Pozzato \cite{olivetti2015standard} and Girlando et al.~\cite{girlando2016},
the difference being that we do not assume the communication rule $\mathsf{com}$.
The basic components of the calculi $\BPCLlessstar$ are Gentzen-style sequents to which is added the following \emph{block structure} from \cite{olivetti2015standard}, representing  $\less$-formulas in the right-hand side of sequents.

\begin{definition}\label{def:hyp}
	A \emph{block} is a structure $\bl{\Sigma}{A}$, where $\Sigma$ is a multiset of formulas and 
	$A$ is a formula. A \emph{sequent with blocks} is a pair $\G \seq \D$, where $\G$  is a multiset of formulas,
	and $\D$  is a multiset of formulas and blocks.
	Sequents are interpreted 
	in $\lan$ as follows
	(where $\D'$ does not contain blocks):
	%
	\begin{center}
		$
		\fint{\G \seq \D', \bl{\Sigma_1}{C_1}, \dots, \bl{\Sigma_k}{C_k}} =$ 
		
		$\AND\G \to \OR\D' \lor (\OR\Sigma_1 \less C_1) \lor \dots \lor (\OR\Sigma_k \less C_k).$
	\end{center}
	A \emph{hypersequent} $\hH$ is a finite multiset of sequents with blocks  
	$\G_1 \seq \D_1 \hyp \dots  \hyp \G_n \seq \D_n$,
	where $\G_1 \seq \D_1, \dots , \G_n \seq \D_n$ are called the \emph{components} of $\hH$. 
	We say that a  hypersequent 
	is \emph{valid in a model} $\M$ if 
	it has a component $\G_k \seq \D_k$ 
	such that $\M\models\fint{\G_k \seq \D_k}$.
\end{definition}

While hypersequents do not have a formula interpretation, sequents with blocks are interpreted as formulas of $\lan$, in a  way different from \cite{olivetti2015standard,girlando2016}. Specifically, for us $\bl{B_1, \dots , B_n}{A}$ is interpreted as $(B_1 \lor \dots \lor B_n) \less A$, while in \cite{olivetti2015standard,girlando2016} it corresponds to $(B_1 \less A) \lor \dots \lor (B_n \less A)$. These two interpretations are equivalent in $\Vless$ but  are not equivalent in $\Nless$.

\begin{figure}
	\begin{adjustbox}{max width = \textwidth}
		\begin{tabular}{|c|}
			
			\hline
			\\[-0.2cm]
			$ \vlinf{\initm}{}{\hG \hyp \G, p \seq p, \D}{} $
			\quad
			$ \vlinf{\lbotm}{}{\hG \hyp \G, \bot \seq \D}{} $
			\quad
			$ \vliinf{\ltom}{}{\hG \hyp \G, A \to B \seq \D}{\hG \hyp \G, \finvh{A \to B }, B \seq \D}{\hG \hyp \G, \finvh{A \to B } \seq \D, A} $
			\\[0.5cm]
			$ \vliinf{\llessm}{}{\hG \hyp \G, A \less B \seq \D, \bl{\Sigma}{C}}{\hG \hyp \G, {A \less B} \seq \D, \bl{B, \Sigma}{C}}{\hG \hyp \G, {A \less B} \seq \D, {\bl{ \Sigma}{C}},\bl{ \Sigma}{A}} $\\[0.5cm]
			$ \vlinf{\rtom}{}{\hG \hyp \G \seq \D, A \to B}{\hG\hyp   \G , \finvh{A \to B }, A \seq \D, B} $
			\ \ 
			$ \vlinf{\rlessm}{}{\hG \hyp \G \seq \D, A \less B}{\hG \hyp \G \seq \D, \finvh{A \less B}, \bl{A}{B}} $
			\ \ 
			$ \vlinf{\jumpm}{}{\hG \hyp \G \seq  \D , \bl{\Sigma}{A}}{\hG \hyp \G \seq \D , \bl{\Sigma}{A} \hyp A \seq \Sigma} $
			\\[0.5cm] 
			$ \vlinf{\hrulenm}{}{\hG \hyp \G \seq \D}{\hG \hyp \G \seq \D, \bl{\bot}{\top}} $
			\qquad\quad
			$ \vliinf{\hruletm}{}{\hG \hyp \G, A \less B \seq \D}{\hG \hyp \G, A \less B \seq \D, B }{\hG \hyp \G, A \less B \seq \D, \bl{\bot}{A}} $\\[0.5cm]
			$ \vlinf{\hrulewm}{}{\hG \hyp \G \seq \D, \bl{\Sigma}{A}}{\hG \hyp \G \seq \D, \bl{\Sigma}{A},\Sigma} $
			\qquad\qquad
			$ \vliinf{\hrulecm}{}{\hG \hyp \G, A \less B \seq \D }{\hG \hyp \G, A \less B \seq \D, B}{ \hG \hyp \G, A \less B , A\seq \D} $\\[0.5cm]
			$ \vlinf{\hruleaml}{}{\hG \hyp \G, A \less B \seq \D \hyp \Omega \seq \Theta}{\hG \hyp \G, A \less B \seq \D \hyp \Omega, A \less B \seq \Theta} $
			\quad \ 
			$ \vlinf{\hruleamr}{}{\hG \hyp \G \seq \D, A \less B  \hyp \Omega \seq \Theta}{\hG \hyp \G \seq \D, A \less B  \hyp \Omega \seq \Theta, A \less B } $ 
			\\
			\hline 
			\multicolumn{1}{c}{}\\[-0.3cm]
			\multicolumn{1}{c}{
				\begin{tabular}{lll}
					$\HNless = \{\initm, \lbotm, \ltom, \rtom\} \cup \{\llessm, \rlessm, \jumpm \}$ &&
					$\HNWless = \HNless \cup \{\hruletm, \hrulewm \}$ \\
					$\HNNless = \HNless \cup \{\hrulenm \}$ &&
					$\HNCless = \HNless \cup \{ \hrulewm, \hrulecm \}$ \\
					$\HNTless = \HNless \cup \{ \hruletm \}$ &&
					$\HNAlessstar = \HNlessstar \cup \{\hruleaml, \hruleamr\}$ \

				\end{tabular}
			}
		\end{tabular}
	\end{adjustbox}
	\caption{\label{figure:HPCLless} Rules of hypersequent calculi $\HNlessstar$.} 
\end{figure}

The calculi $\HNlessstar$ are defined in Fig.~\ref{figure:HPCLless}. 
The rules are 
\emph{cumulative}, 
meaning that each rule has the 
principal formula 
copied in the premisses. 
Differently from the calculi $\SNlessstar$ in the previous section, 
$\BPCLlessstar$ have separate left and right rules
for $\less$, and all rules have a fixed number of premisses.

We point out that the hypersequent structure is not necessary to define sequent calculi with blocks for \CPN\ logics.
Moreover, it can be checked that a
hypersequent  is derivable if and only if one of its components is derivable. 
Following the strategy from \cite{dalmonte2021hypersequent}, 
we chose to employ a hypersequential structure to obtain invertibility of \emph{all} the rules of the calculi, there including the $ \jumpm $ rule, which was not invertible in \cite{girlando2016}. 
Together with their cumulative formulation, invertibility of the rules allows to directly construct countermodels from failed proof search, without the need of backtracking or inserting any additional computation. 
Moreover, differently from \cite{girlando2016}, the countermodel construction modularly extends to logics with \emph{absoluteness}. The rules for absoluteness are inspired from \cite{girlando2017hypersequent}, and  correspond to condition $ \mathsf{A}+ $ from Sec.~\ref{sec:semantics}. 
Soundness of the rules 
is proved as follows. 
Let $ \BPCLlessstar \vd A$ denote derivability of $ A $ in $ \BPCLlessstar $.

\begin{theorem}[Soundness]\label{th:soundness hyp}
	For every formula $A$, if $A$ is derivable in $\BPCLlessstar$, then $A$ is valid in all $\Nlessstar$-models.
\end{theorem}
\begin{proof}
	For every rule $R$ of $\HNlessstar$, we show that if the premisses of $R$ are valid in a $\Nlessstar$-model $\M$,
	then the conclusion is also valid in $\M$. 
	We only consider 
	some relevant examples of modal rules.
	($\llessm$) 
	Suppose  $\M \models \hG \hyp \G, {A \less B} \seq \D, \bl{B, \Sigma}{C}$
	and
	$\M\models \hG \hyp \G, {A \less B} \seq \D, {\bl{ \Sigma}{C}},\bl{ \Sigma}{A}$.
	If $\M\models \hG$ we are done.
	Otherwise $\M \models \AND\G \land (A \less B) \to (B\lor \OR\Sigma \less C) \lor \OR\D$
	and
	$\M\models \AND\G \land (A \less B) \to (\OR\Sigma \less C) \lor (\OR\Sigma \less A) \lor \OR\D$.
	Then by $\axTR$, 
	$\M\models \AND\G \land (A \less B) \to (\OR\Sigma \less C) \lor (\OR\Sigma \less B) \lor \OR\D$,
	and by $\CPR$ and $\axOR$,
	$\M\models \AND\G \land (A \less B) \to (\OR\Sigma \less C) \lor (\OR\Sigma \less B \lor \OR\Sigma) \lor \OR\D$,
	therefore by $\axTR$, 
	$\M\models \AND\G \land (A \less B) \to (\OR\Sigma \less C) \lor (\OR\Sigma \less C) \lor \OR\D$, thus
	$\M\models \AND\G \land (A \less B) \to (\OR\Sigma \less C) \lor \OR\D$.
	It follows $\M \models \hG \hyp \G, A \less B \seq \D, \bl{\Sigma}{C}$.
	%
	%
	($\jumpm$)
	Suppose that $\M\models \hG \hyp \G \seq \D , \bl{\Sigma}{A} \hyp A \seq \Sigma$.
	If $\M\models \hG \hyp \G \seq \D , \bl{\Sigma}{A}$ we are done, 
	otherwise $\M\models A \to \OR\Sigma$.
	Then by $\CPR$, $\M \models \OR\Sigma \less A$, 
	therefore $\M \models \G \seq \D , \bl{\Sigma}{A}$.
	($\hruletm$)
	Suppose that $\M\models \hG \hyp \G, A \less B \seq \D, B$ and $\M\models \hG \hyp \G, A \less B \seq \D, \bl{\bot}{A}$.
	If $\M\models \hG$ we are done, otherwise 
	$\M\models \AND\G\land (A \less B) \to \OR\D\lor (B \land (\bot\less A))$.
	Then by $\axTR$, $\M\models \AND\G\land (A \less B) \to \OR\D\lor (B \land (\bot\less B))$,
	and by $\axT$, $\M\models \AND\G\land (A \less B) \to \OR\D\lor (B \land\neg B)$,
	therefore $\M\models \AND\G\land (A \less B) \to \OR\D$.
\end{proof}

The calculi  $\BPCLlessstar $ enjoy admissibility of the following structural properties:


\begin{lemma}
	It holds that all the rules of $\BPCLlessstar $ are 
	height-preserving
	invertible, and that the following rules of weakening and contraction are 
	height-preserving
	admissible in $ \BPCLlessstar $, where $ A $ in $ \wkrm $ or $\ctrr$ can be a formula or a block. 
	\begin{center}\begin{small}$
			\vlinf{\wklm}{}{\hG \hyp A, \G \seq \D }{\hG \hyp \G \seq \D}
			\quad 
			\vlinf{\wkrm}{}{ \hG \hyp \G \seq \D, A }{\hG \hyp \G \seq \D}
			\quad 
			\vlinf{\wkc}{}{ \hG \hyp \mathcal{H} }{\hG}
			\quad 
			\vlinf{\wkb}{}{ \hG \hyp \G \seq \D, \bl{\Sigma, B}{A} }{\hG \hyp \G \seq \D, \bl{\Sigma}{A}}$
			
			\vspace{0.2cm}
			$\vlinf{\ctrl}{}{\hG \hyp A, \G \seq \D }{\hG \hyp A, A, \G \seq \D}
			\quad 
			\vlinf{\ctrr}{}{ \hG \hyp \G \seq \D, A }{\hG \hyp \G \seq \D, A, A}
			\quad 
			\vlinf{\ctrc}{}{ \hG \hyp \mathcal{H} }{\hG \hyp \mathcal{H} \hyp \mathcal{H}}
			\quad 
			\vlinf{\ctrb}{}{ \hG \hyp \G \seq \D, \bl{\Sigma, B}{A} }{\hG \hyp \G \seq \D, \bl{\Sigma, B, B}{A}}
			$\end{small}\end{center}
	
\end{lemma}

\begin{proof}
	Height-preserving admissibility of weakening can be standardly proved by induction on the height of the derivation. Invertibility of all the rules of $\BPCLlessstar $ immediately follows. For instance, the premiss of the $\jumpm$ rule can be derived from the conclusion of $ \jumpm $ using $ \wkc $. Admissibility of contraction also follows by standard induction on the height of derivations. 
\end{proof}

Concerning completeness, a proof can be given by showing that
the derivations in the calculi $\SNlessstar$ can be simulated in $\BPCLlessstar$. 


\begin{theorem}[Simulation]
	\label{thm:simulation_manypremiss}
	For $ A  $ $ \lan $ formula, if $ \SNlessstar \vdash A $ then $ \BPCLlessstar \vdash A$. 
\end{theorem}
\begin{proof}
	We show that the rules of $ \SNlessstar $ can be stepwise simulated by the rules of $ \BPCLlessstar $. 
	Then the proof of the claim is similar to the one given in \cite{girlando2016}. 
	Let $ \manyless =  C_1 \less D_1, ..., C_n \less D_n $ and $ \manyd{1}{n} = D_1, \dots, D_n $. 
	Here follows the translation of $ \seqruleCPNm_n $. 
	
	\begin{adjustbox}{max width = \textwidth}
		$
		\vlderivation{
			\vlin{\rlessm}{}{ \G, \manyless  \seq A \less B, \D }{
				\vliin{\llessm}{}{\G, \manyless  \seq \bl{A}{B},  \D}{
					\vlin{\llessm}{}{\G,\manyless \seq \bl{A, D_1}{B},  \D}{	
						\vliin{\llessm}{}{
							\begin{matrix}
							\G,\manyless \seq \bl{A, \manyd{n-1}}{B},  \D\\
							\vdots\\
							\G,\manyless \seq \bl{A, D_1, D_{2}}{B},  \D
							\end{matrix}
						}{
							\vlin{\jumpm}{}{\G,\manyless \seq \bl{A, \manyd{1}{n}}{B},  \D}{
								\vlin{\wkc}{}{\G,\manyless \seq \bl{A,\manyd{1}{n}}{B},  \D \hyp  B \seq A, \manyd{1}{n}}{	\vlhy{B \seq A, \manyd{1}{n}}}
							}
						}{
							\vlin{\jumpm}{}{\G,\manyless \seq \bl{A, \manyd{1}{n-1}}{C_n},  \D}{
								\vlin{\wkc}{}{\G,\manyless \seq \bl{A, \manyd{1}{n-1}}{C_n},  \D \hyp C_n \seq A, \manyd{1}{n-1}  }{\vlhy{C_n \seq A, \manyd{1}{n-1}}}
							}
						}
					}
				}
				{
					\vlin{\jumpm}{}{\G, \manyless \seq \bl{A}{C_1},  \D}{
						\vlin{\wkc}{}{ \G, \manyless \seq \bl{A}{C_1},  \D \hyp C_1 \seq A}{\vlhy{C_1 \seq A}}
					}
				}
			}
		}
		$
	\end{adjustbox}
	
	\vspace{0.5cm}
	
	\noindent 
	Rule $ \seqruleNm_n $ is derived in a similar way, by replacing $ \rlessm $ with $ \hrulenm $ and removing occurrences of $ A $ and $ B $ in the derivation above. For $ \seqruleAm_n $ and $ \seqruleNAm_n $ each $ \jumpm $ is followed by applications of $ \hruleaml $ and $ \hruleamr $.  
	To derive rule $ \seqruleWm_n $, replace the upper leftmost occurrence of $ \jumpm $ with rule $ \hrulewm $.
	Rule $ \seqruleWtwom $ is immediately derivable using $ \jumpm $ and $ \hrulewm $, and rule $ \seqruleCm $ using $ \hrulecm $.   
	The case of rule $ \seqruleTm_n $ is more complex. 
	We start with the following derivation.

	\begin{adjustbox}{max width = \textwidth}
		$ \vlderivation{
			\vliin{\hruletm}{}{ 
				\G, \manyless \seq \D
			}{
				\vlin{\hruletm}{}{
					\G, \manyless \seq \D, {D_n}
				}{
					\vliin{\hruletm}{}{
						\begin{matrix}
						\G, \manyless \seq \D,\manyd{2}{n} \\
						\vdots\\
						\G, \manyless \seq \D, D_n, {D_{n-1}}\\
						\end{matrix}
					}{
						\vlhy{\G, \manyless \seq \D, \manyd{2}{n}, {D_1} }
					}{
						\vlin{\jumpm}{}{ \G, \manyless \seq \D,\manyd{2}{n}, \bl{\bot}{{C_1}} }{
							\vlin{\wkb,\wkm}{}{ \G, \manyless \seq \D,\manyd{2}{n}, \bl{\bot}{{C_1}} \hyp C_1 \seq \bot}{\vlhy{C_1 \seq }}
						}
					}
				}
			}{\vlhy{\G, \manyless \seq \D, \bl{\bot}{{C_n}} }}
		} $
	\end{adjustbox}
	
	\ 
	
	\noindent The leftmost sequent is premiss $ \G, \manyless \seq \D, \manyd{1}{n} $ of $ \seqruleTm_n $. We now construct from the remaining premisses of $ \seqruleTm_n $, that is, $ \{C_k \seq D_1 , \dots , D_{k-1}\} $, for $ k \leq n $,  
	derivations of sequents  $ \G, \manyless \seq \D, \manyds, \bl{\bot}{C_{k}} $, for $ 1 < k \leq n $, where $ \manyds = D_n, \dots, D_{k+1} $ if $ k < n $, and is empty otherwise. In applications of the $ \jumpm $ rule, we omit specifying the leftmost component of the hypersequents.
	
	\
	
	\begin{adjustbox}{max width = \textwidth}
		$
		\vlderivation{
			\vlin{\llessm}{}{\G, \manyless \seq \D, \manyds, \bl{\bot}{C_{k}}}
			{
				\vliin{\llessm}{}{ \G, \manyless \seq \D, \manyds, \bl{ {D_1}, \bot}{C_{k}}  }{
					\vlin{\llessm}{}{ \G, \manyless \seq \D, \manyds, \bl{D_1, {D_{2}}, \bot}{C_{k}} }{
						\vliin{\llessm}{}{ 
							\begin{matrix}
							\G, \manyless \seq \D, \manyds, \bl{ \manyd{1}{k-2} , \bot}{C_{k}}\\
							\vdots\\
							\G, \manyless \seq \D, \manyds, \bl{  D_1, D_2, {D_{3}}, \bot}{C_{k}}
							\end{matrix}
						}{ 
							\vlin{\jumpm}{}{\G, \manyless \seq \D, \manyds, \bl{ \manyd{1}{k-1} , \bot}{C_{k}}}{
								\vlin{\wkb,\wkm}{}{ \dots \hyp C_k \seq \manyd{1}{k-1}, \bot }{
									\vlhy{C_k \seq \manyd{1}{k-1} }}
							}
						}{ 
							\vlin{\jumpm}{}{ 
								\G, \manyless \seq \D, \manyds, \bl{ \manyd{1}{k-2} , \bot}{C_{k}  },\bl{ \manyd{1}{k-2} , \bot}{ C_{k-1} }
							}{
								\vlin{\wkc,\wkm}{}{ \dots \hyp C_{k-1} \seq  \manyd{1}{k-2} , \bot}{
									\vlhy{C_{k-1} \seq  \manyd{1}{k-2}  }
								}
							}
						}
					}
				}{ 
					\vlin{\jumpm}{}{ \G, \manyless \seq \D, \manyds, \bl{D_1, \bot}{C_{k}}, \bl{D_1, \bot}{{C_{2}}} }{
						\vlin{\wkc,\wkm}{}{\dots \hyp C_2 \Rightarrow D_1, \bot }{ \vlhy{ C_2 \Rightarrow D_1} }
					} 
				}
			}
		}
		$
	\end{adjustbox}

	\
	
	\noindent The rightmost premiss of the lower occurrence of $ \llessm $, not shown, is sequent $ \G, \manyless \seq \D, \manyds, \bl{\bot}{C_{k}}, \bl{\bot}{C_{1}} $, which is derivable by $ \jumpm $ from premiss $ C_1 \Rightarrow $. 
\end{proof}

Since $\SNlessstar$ are complete with respect to $\Nlessstar$, this simulation
entails that $ \BPCLlessstar$ are also complete.
Here we present in more detail an alternative completeness proof based on the semantics.
In particular, we define a terminating bottom-up proof-search strategy in $\BPCLlessstar$,
and show that whenever the strategy fails, one can directly extract a countermodel
of the root formula/hypersequent.
The strategy is based on the following notion of saturation.


\begin{definition}
	Let $\hH = \G_1 \seq \D_1 \hyp ... \hyp \G_n \seq \D_n$ 
	be a hypersequent occurring in proof for $\hH'$ in $\BPCLlessstar$.
	The \emph{saturation conditions} associated to each application of a rule of $\BPCLlessstar$ are as follows:
	($\initm$) $\G_k\cap\D_k = \emptyset$.
	($\lbotm$) $\bot\notin\G_k$.
	($\ltom$) If $A\to B\in\G_k$, then $A\in\D_k$ or $B\in\G_k$.
	($\rtom$) If $A\to B\in\D_k$, then $A\in\G_k$ and $B\in\D_k$.
	($\llessm$)  If $A \less B \in \G_k$ and $\bl{\Sigma}{C}\in\D_k$, then $B \in\Sigma$ or 
	there is $\bl{\Pi}{A}\in\D_k$ such that $\set\Sigma\subseteq\set\Pi$.
	($\rlessm$) If $A \less B \in \D_k$, then there is $\bl{\Sigma}{B}\in\D_k$ such that $A\in\Sigma$. 
	($\jumpm$) If $\bl{\Sigma}{A}\in\D_k$, then there is $\G_j \seq \D_j\in\hH$ such that $A\in\G_j$ and $\set\Sigma\subseteq\D_j$.
	($\hrulenm$) 
	There is $\bl{\Sigma}{\top}\in\D_k$ such that $\bot\in\Sigma$.         
	($\hruletm$) If $A\less B\in\G_k$, then $B\in\D_k$ or 
	there is $\bl{\Sigma, \bot}{A}\in\D_k$.
	($\hrulewm$) If $\bl{\Sigma}{A}\in\D_k$, then $\set\Sigma\subseteq\D_k$.
	($\hrulecm$) If $A\less B\in\G_k$, then $B\in\D_k$ or $A\in\G_k$.
	($\hruleaml$) If $A\less B\in\G_k$, 
	then for all $\G_j \seq \D_j\in\hH$, $A \less B\in\G_j$.
	($\hruleamr$) If $A \less B \in\D_k$, 
	then for all $\G_j \seq \D_j\in\hH$, $ A \less B  \in\D_j$.
	%
	We say that $\hH$ is \emph{saturated with respect to an application of a rule} $R$ if it satisfies the saturation condition ($R$)
	for that particular rule application, and it is \emph{saturated with respect to} $\BPCLlessstar$ if it is saturated with respect to all possible applications of any rule of $\BPCLlessstar$.
\end{definition}

The strategy consists simply in applying the rules backward
until no additional rule application is possible respecting the following two conditions:
(i) no rule can be applied to an initial hypersequent;
(ii) the application of a rule is not allowed if the hypersequent is already saturated with respect to that specific
rule application.
The conditions (i) and (ii) ensure that proof-search terminates for every hypersequent $\hH$.

\begin{proposition}
	Proof-search for $\hH$ in $\BPCLlessstar$ in accordance with the strategy always terminates after 
	a finite number of steps.
\end{proposition}
\begin{proof}
	Let $\mathscr P$ be a proof of $\hH$ constructed according to the strategy. 
	Then all formulas occurring in $\mathscr P$ (both inside and outside blocks) 
	are subformulas of formulas of $\hH$ or they are $\bot$ or $\top$, so they are finitely many. 
	Moreover, the saturation conditions prevent duplications of the same formulas
	(both inside and outside blocks) and the same blocks. 
	It follows that all hypersequents occurring in  $\mathscr P$ have a finite length, moreover every branch
	of  $\mathscr P$ contains only finitely many hypersequents. 
\end{proof}

If the strategy succeeds, then it constructs a derivation of the root hypersequent $\hH$. Otherwise, a saturated hypersequent will occur in the leaf of a branch. 
We now prove that the proof-search strategy is complete, showing that whenever the strategy fails,
from every saturated hypersequent one can directly construct a countermodel for $\hH$.

\begin{proposition}[Countermodel construction]\label{prop:countermodel}
	Let $\hH = \G_1 \seq \D_1 \hyp \dots \hyp \G_k \seq \D_k$ 
	be a saturated hypersequent occurring in a proof search tree for $\hH_0$ in $\BPCLlessstar$ built in accordance with the 
	strategy. 
	For $\Sigma$ multiset of formulas, let  
	$\SDelta = \{n \mid \G_n \seq \D_n \in \hH \textup{ and } \set\Sigma \subseteq \D_n\}$.
	%
	We define 
	$\M= \langle \W, \N, \V\rangle$:
	
	\begin{itemize}
		\item $\W = \{n \mid \G_n \seq \D_n \in \hH\}$. 
		\item For every $n\in\W$, $\N(n) = \{\SDelta \mid \textup{there is } A \textup{ such that } \bl{\Sigma}{A}\in\D_n\}$.
		\item For every $p\in\atm$, $\V(p) = \{n\in\W \mid p\in\G_n\}$.
	\end{itemize}
	Then for all $n \in \W$, 
	(i) if $A\in\G_n$, then $n \Vd A$;  
	(ii) if $A\in\D_n$, then $n \not\Vd A$; and 
	(iii) if $\bl{\Sigma}{A}\in\D_n$, then $n\not\Vd \OR\Sigma \less A$.
	Moreover $\M$ is a $\Nlessstar$-model.
\end{proposition}
\begin{proof}
	The claims (i), (ii) and (iii) are proved simultaneously by induction on
	the following notion of complexity of formulas and blocks:
	$c(p) = c(\bot) = 1$, $c(A \to B) = c(A \less B) = c(A) + c(B) + 1$, 
	$c(\bl{B_1, ..., B_n}{A}) = c(B_1) + ... + c(B_n) + c(A)$.
	For $A = p, \bot, B \to C$ the proof is routine.
	We consider the case $A = B \less C$.
	($B \less C \in \G_n$)
	Suppose $\alpha\in\N(n)$. 
	By definition, $\alpha = \SDelta$ for some $\Sigma$
	such that there is $\bl{\Sigma}{D}\in\D_n$.
	Then by saturation of $\llessm$,
	$C \in \Sigma$ or there is $\Pi$ such that $\set\Sigma\subseteq\set\Pi$ and $\bl{\Pi}{B}\in\D_n$.
	In the first case, 
	for every $m\in\SDelta$, $C \in\D_m$, then by \ih, $m \not\Vd C$.
	Therefore $\SDelta\not\EVd C$.
	In the second case,
	by saturation of $\jumpm$ there is $m\in\W$ such that $B \in\G_m$ and $\set\Pi\subseteq\D_m$, 
	thus $\set\Sigma\subseteq\D_m$ .
	Then by \ih, $m \Vd B$, and by definition $m\in\SDelta$. 
	Therefore $\SDelta\EVd B$.
	It follows $n\Vd B \less C$.
	($B \less C \in \D_n$)
	By saturation of $\rlessm$,
	there is $\bl{\Sigma}{C}\in\D_n$ such that $B\in\Sigma$.
	Then by definition, $\SDelta \in\N(n)$,
	and by \ih, $m\not\Vd B$ for every $m\in\SDelta$,
	that is $\SDelta\not\EVd B$.
	Moreover, by saturation of $\jumpm$ there is $m\in\W$ such that
	$\set\Sigma\subseteq\D_m$
	and $C \in\G_m$.
	Then by \ih, $m \Vd C$, and by definition $m\in\SDelta$,
	thus $\SDelta\EVd C$.
	Therefore $n\not\Vd B \less C$.
	($\bl{\Sigma}{A} \in \D_n$)
	Analogous to the previous item, considering that by \ih\ $m\not\Vd B$ for all $m\in\SDelta$ and $B\in\Sigma$,
	that is $\SDelta\not\EVd \OR\Sigma$.
	
	We now show that $\M$ satisfies the conditions of $\Nlessstar$-models.
	(Non-emptyness) If $\alpha\in\N(n)$, then $\alpha = \SDelta$ for some $\Sigma$ such that
	there is $\bl{\Sigma}{A}\in\D_n$.
	Then by saturation of $\jumpm$,
	there is $m\in\W$ such that $A\in\G_m$ and $\set\Sigma\subseteq\D_m$,
	thus $m\in\SDelta$.
	(Normality) By saturation of $\hrulenm$, there is $\bl{\Sigma,\bot}{\top}\in\D_n$,
	thus $(\Sigma,\bot)^\Delta\in\N(n)$, that is $\N(n)\not=\emptyset$.
	(Total reflexivity)
	We modify the definition of the neighbourhood function as follows. 
	For all $n\in\W$, let $ \ON(n) = \bigcup \N(n) \cup \{n\}$.
	Then, define $ \NT(n) = \N(n) \cup \ON(n)$.   
	We show that 
	the claim (i) above still holds ((ii) and (iii) are proved as before):
	Suppose $B\less C\in\D_n$.
	As before we can prove that $\Sigma^\Delta\not\EVd C$ or $\Sigma^\Delta\EVd B$ for all $\bl{\Sigma}{A}\in\D_n$.
	Here we show that the same holds for $\ON(n)$.
	If there is $\alpha \in \N(n)$ such that $\alpha\EVd B$, then $\ON(n)\EVd B$. 
	If instead there is no $\alpha \in \N(n)$ such that $\alpha\EVd B$,
	then $\alpha\not\EVd C$ for all $\alpha \in \N(n)$, that is $\bigcup \N(n)\not\EVd C$.
	Assume by contradiction that $\ON(n)\EVd C$. Then $n\Vd C$.
	Moreover by saturation of $\hruletm$, $C\in\D_n$ 
	or there is $\bl{\Pi}{B}\in\D_n$ such that $\bot\in\Pi$.
	If $C\in\D_m$, then by \ih, $n\not\Vd C$, contradicting $n\Vd C$. 
	If $\bl{\Pi}{B}\in\D_n$,
	then by saturation of $\jumpm$ there is $m\in\W$ such that $B\in\G_m$ and $\set\Pi\subseteq\D_m$.
	Then by \ih, $m\Vd B$, moreover $\Pi^\D\in\N(n)$ and $m\in\Pi^\D$, 
	thus $\Pi^\D\EVd B$, against the hypothesis.
	Therefore $\ON(n)\not\EVd C$.
	(Weak centering)  If $\alpha\in\N(n)$, then $\alpha = \SDelta$ for some $\Sigma$ such that
	there is $\bl{\Sigma}{A}\in\D_n$.
	Then by saturation of $\hrulewm$, $\set\Sigma\subseteq\D_n$, thus $n\in\SDelta$.
	(Centering) 
	We modify the definition of 
	the neighbourhood function as 
	$\NC(n) = \N(n) \cup \{\{n\}\}$.
	We show that 
	(i) still holds ((ii) and (iii) are as before):
	Suppose $B\less C\in\D_n$.
	As before we can prove that $\Sigma^\Delta\not\EVd C$ or $\Sigma^\Delta\EVd B$ for all $\bl{\Sigma}{A}\in\D_n$.
	Here we show that the same holds for $\{n\}$.
	By saturation of $\hrulecm$, $C\in\D_n$ or $B\in\G_n$, 
	thus by \ih, $m\Vd B$ or $m\not\Vd C$, therefore $\{n\}^\D\EVd B$ or $\{n\}\not\EVd C$.
	(Strong absoluteness) 
	We modify the definition of 
	$\N$
	as 
	$\NA(n) = \{\SDelta \mid \textup{there are } m\in\W \textup{ and } A \textup{ such that } \bl{\Sigma}{A}\in\D_m\}$.
	We show that (i) still 
	holds ((ii) and (iii) are as before).
	Suppose $B \less C \in \G_n$ and $\alpha\in\N(n)$. 
	Then $\alpha = \SDelta$ for some $\bl{\Sigma}{D}\in\D_m$ for some $m\in\W$.
	By 
	saturation of $\hruleaml$, $B \less C \in \G_m$,
	then by saturation of $\llessm$,
	$C \in \Sigma$ or there is $\Pi$ such that $\set\Sigma\subseteq\set\Pi$ and $\bl{\Pi}{B}\in\D_m$.
	In the first case, 
	$\SDelta\not\EVd C$.
	In the second case,
	by saturation of $\jumpm$ there is $k\in\W$ such that $B \in\G_k$ and $\set\Pi\subseteq\D_k$, 
	thus $\set\Sigma\subseteq\D_m$, therefore $\SDelta\EVd B$.
\end{proof}

Note that, since all rules are cumulative, the claims (i) and (ii) of Prop.~\ref{prop:countermodel} also hold for  the root hypersequent $\hH_0$, thus $\M$ is a countermodel of $\hH_0$.
Moreover, since every proof built in accordance with the strategy either provides a derivation of
the root hypersequent, or contains a saturated hypersequent, 
this result entails 
a constructive proof of 
the completeness of $\BPCLlessstar$.

\begin{theorem}[Semantic completeness]
	For every hypersequent $\hH$, if $\hH$ is valid in all $\Nlessstar$-models,
	then $\hH$ is derivable in $\BPCLlessstar$.
\end{theorem}

Here follows an 
example of the countermodel construction. 

\begin{example}
	We show that axiom $\axCO$ is not derivable in $\Nless$.
	Here follows a failed proof of $\seq (p \less q) \lor (q \less p)$ in $\HNless$, where $ \hH $ is saturated, and $ \vlor_R $ is admissible from the rules of $\Nless$: 
	$$
	\vlderivation{
		\vlin{\vlor_\mathsf{R}}{}{\seq (p \less q) \lor (q \less p)}{
			\vlin{\rlessm(\times 2)}{}{\seq p \less q, q \less p, (p \less q) \lor (q \less p)}{
				\vlin{\jumpm (\times 2)}{}{\seq \bl{q}{p}, \bl{p}{q}, p \less q, q \less p, (p \less q) \lor (q \less p)}{
					\vlhy{\hH: \ \ \seq \bl{q}{p}, \bl{p}{q}, p \less q, q \less p, (p \less q) \lor (q \less p) \hyp q \seq p \hyp p \seq q }
				}
			}
		}
	}
	$$
	We consider the following enumeration of the components of the saturated hypersequent $\hH$:
	1: $\seq \bl{q}{p}, \bl{p}{q}, p \less q, q \less p, (p \less q) \lor (q \less p)$;
	2: $q \seq p$; and
	3: $p \seq q$.
	Then, following 
	the construction of Prop.~\ref{prop:countermodel}
	we obtain the 
	following countermodel $\M = \langle \W, \N, \V\rangle$:
	%
	$\W = \{1, 2, 3\}$.
	$\N(1) = \{\Deltaset{p}, \Deltaset{q}\} = \{\{2\}, \{3\}\}$, and $\N(2) = \N(3) = \emptyset$.
	$\V(p) = \{3\}$ and $\V(q) = \{2\}$.
	Then we have $\Deltaset{p}\EVd q$ and $\Deltaset{p}\not\EVd p$, thus $1\not\Vd p\less q$, 
	moreover $\Deltaset{q}\EVd p$ and $\Deltaset{q}\not\EVd q$, thus $1\not\Vd q\less p$.
	Therefore $1 \not\Vd (p \less q) \lor (q \less p)$.
\end{example}

\section{Conclusions}
\label{sec:rel works}

We 
introduced \CPN\ logics, which are a generalisation of Lewis' logics 
of comparative plausibility
defined over neighbourhood rather than sphere models. 
As a difference with sphere models, neighbourhoods need not to be nested, allowing to express
more
general notions of comparative plausibility. 
From a proof-theoretic viewpoint,
\CPN\ logics are 
captured by 
suitable restrictions of 
sequent calculi for Lewis' logics:
they coincide to restrictions of calculi from \cite{lellmann:phd,lellpatt} to a single principal $\less$-formula
in the right-hand side of sequents, 
and  
the  single-component formulation of their hypersequent calculi
corresponds to the structured calculi from \cite{olivetti2015standard,girlando2016} without the communication rule.

Overall, 
\CPN\ logics 
represent  a general theory of comparative plausibility with well-understood proof theory and semantics. 
Differently from stronger logics expressing comparative plausibility, \CPN\ logics allow to model preference or similarity in situations where no priority order is assumed between states of affairs or concepts. 
Moreover, \CPN\ logics are an expressive framework, encompassing Lewis' logics \cite{lewis}, which are obtained by adding nesting to \CPN\ logics. In future work we plan to investigate the relations between CPN logics and other well-known comparative plausibility logics introduced in the literature, most notably Halpern's comparative plausibility logics defined over preferential structures \cite{halpern:1997}. We conjecture that Halpern's logics could be obtained by adding the property of \emph{closure under non-empty intersections} to neighbourhood models, which is required to prove equivalence between neighbourhood and preferential structures. Moreover, we wish to relate our systems with the logic of comparative obligation 
introduced by Brown \cite{brown:1996}. 
Brown's operator
is defined on a kind of neighbourhood models
containing a function $\R: \W \longrightarrow \pow(\pow(\pow(\W)))$, 
representing a degree of urgency of obligation.

Furthermore, \CPN\ logics parallel the preferential conditional logics studied in \cite{burgess:1981}. These logics generalise Lewis' counterfactual logics, and admit a neighbourhood semantics, introduced in~\cite{girlando:2021}.  Interestingly, while comparative plausibility and conditional entailment are interdefinable 
in sphere models, the two operators are not interdefinable in neighbourhood semantics, giving rise to two independent theories. While in ~\cite{girlando:2021} a proof-theoretical analysis of the conditional operator in neighbourhood semantics is proposed,  this work explores the behaviour of the comparative plausibility operator in neighbourhood structures. Moreover, having lost the interdefinability between $\less$ and $\cond$,
we wish to study whether alternative and meaningful notions of conditional entailment can be defined in terms
of comparative plausibility. 
We also intend to study applications of \CPN\ logics,
possibly related to the analysis of information sources.

Concerning the proof theory for \CPN\ logics, we wish to analyse the complexity of the logics based on the decision procedure induced by the multi-premisses and the hypersequent calculi. 
Moreover, we plan to automate the proof search and countermodel construction of the hypersequent calculi within a theorem prover,
along the lines of what done in \cite{dalmonte2020hypno,girlando2017vinte,girlando2022calculi}.
We will also investigate extensions of the hypersequent calculi to \CPN\ logics with uniformity, possibly adapting the approach  proposed in \cite{girlando2017hypersequent} for Lewis' logics to our setting, as well as with other semantic conditions, aiming at developing a uniform proof-theoretic account of  \CPN\ logics.

\vspace{0.3cm}

\noindent \textbf{Acknowledgements.} We wish to thank Bj\"orn Lellmann for his suggestions
and contributions to the analysis of the comparative plausibility operator.

\bibliographystyle{aiml22}

\end{document}